\documentclass[12pt]{amsart}

\usepackage{amssymb,amsthm,amsmath}
\usepackage[numbers,sort&compress]{natbib}
\usepackage{color}
\usepackage{graphicx}
\usepackage{tikz}

\title[ ]{  Small  denominators and  large numerators of quasiperiodic Schr\"odinger  operators }

\author{Wencai Liu}
\address[W. Liu]{ Department of Mathematics, Texas A\&M University, College Station, TX 77843-3368, USA} \email{liuwencai1226@gmail.com; wencail@tamu.edu}

\newcommand{\R}{\mathbb{R}}

\newcommand{\Z}{\mathbb{Z}}

\newcommand{\Q}{\mathbb{Q}}

\theoremstyle{plain}
\newtheorem{theorem}{Theorem}[section]

\newtheorem{lemma}[theorem]{Lemma}

\theoremstyle{definition}

\begin{document}

	%%% ----------------------------------------------------------------------
	
	\begin{abstract}
		%	Study of spectral properties of quasiperiodic operators involves dealing with small denominator problems. Both non-perturbative and perturbative approaches have been seen remarkable progress in the last decades. 
		%	In this paper, we develop a method to treat   numerators  and denominators simultaneously, which allows us to obtain sharp estimates of Green's functions arising from quasi-periodic operators. In particular, we can study the completely resonant phases of the almost Mathieu operators.
	%	We initiate an approach to simultaneously treat numerators and denominators arising from quasi-periodic Schr\"odinger operators, which allows us to obtain (possibly) sharp estimates of Green's functions. In particular, we
	%	can study  completely resonant phases of the almost Mathieu operator.
		
	 	We initiate an approach to simultaneously treat numerators and denominators of Green's functions arising from quasi-periodic Schr\"odinger operators, which in particular  allows us to  study  completely resonant phases of the almost Mathieu operator.
		
		Let  $ (H_{\lambda,\alpha,\theta}u) (n)=u(n+1)+u(n-1)+  2\lambda \cos2\pi(\theta+n\alpha)u(n)$ be
		the  almost Mathieu operator on $\ell^2(\mathbb{Z})$, where $\lambda, \alpha, \theta\in \mathbb{R}$.
   Let   $$	\beta(\alpha)=\limsup_{k\rightarrow \infty}-\frac{\ln ||k\alpha||_{\mathbb{R}/\mathbb{Z}}}{|k|}.$$
		We prove that   for any $\theta$ with $2\theta\in \alpha \mathbb{Z}+\mathbb{Z}$, 
		$H_{\lambda,\alpha,\theta}$
		satisfies Anderson localization if $|\lambda|>e^{2\beta(\alpha)}$.
		This confirms a conjecture of 	Avila and Jitomirskaya [The Ten Martini Problem. Ann. of Math. (2) 170 (2009), no. 1, 303--342] and a particular case of a conjecture of Jitomirskaya [Almost everything about the almost Mathieu operator. II. XIth International Congress of Mathematical Physics (Paris, 1994), 373--382, Int. Press, Cambridge, MA, 1995].

	\end{abstract}
	%%% ----------------------------------------------------------------------
	\maketitle
	%%% ----------------------------------------------------------------------
	%%%%%%%%%%%%%%%%%%%%%%%%%%%%%%%%%%%%%%%%%%%%%%%%%%%%%%%%%%%%%%%%%%%%%%%%%%
	%INTRODUCTION%%%%%%%%%%%%%%%%%%%%%%%%%%%%%%%%%%%%%%%%%%%%%%%%%%%%%%%%%%%%%
	%%%%%%%%%%%%%%%%%%%%%%%%%%%%%%%%%%%%%%%%%%%%%%%%%%%%%%%%%%%%%%%%%%%%%%%%%%
	
	\section{Introduction}
	In this paper, we   study  one-dimensional quasiperiodic Schr\"odinger operators 
	$H = H_{\lambda v,\alpha,\theta}$ defined on  $\ell^2(\Z)$:
	\begin{equation}\label{op}
		(H_{\lambda v,\alpha,\theta}u)(n)=u({n+1})+u({n-1})+ \lambda  v(\theta+n\alpha)u(n),  % \text{ with }  v(\theta)=2\cos2\pi \theta,
	\end{equation}
	where $v : \R/\Z\to \R$ is the potential, $\lambda $ is the coupling constant, $\alpha \in \R  \backslash \Q $ is the frequency, and $\theta \in \R$ is the phase. 
	One of the  most studied  examples in both mathematics and physics  is  the almost Mathieu operator (denoted by $H_{\lambda,\alpha,\theta}$), where  $v(\theta) = 2 \cos(2 \pi \theta).$

	There are two small   denominator problems arising from quasi-periodic  Schr\"odinger operators. 
	In the regime of small  coupling constants, the  operator  is close to the   free  discrete Schr\"odinger operator and thus,  it is natural to  establish  the reducibility to a constant cocycle by writing down the eigen-equation $Hu=Eu$ as a dynamical system on $ (\R/\Z,\R^2)$: 
	\begin{equation}\label{g122}
		(\theta,w) \to (\theta+\alpha, A(\theta) w ),
	\end{equation}
	where $A(\theta)=\left(
	\begin{array}{cc}
		E-\lambda v (\theta )& -1 \\
		1& 0\\
	\end{array}
	\right)$ 
	is referred  to as Schr\"odinger cocycle of  ~\eqref{op}. 
	In the regime of large coupling constants, the operator  can be viewed as a perturbation of a purely  diagonal matrix with dense eigenvalues. 
	In both regimes, 
  spectral theory of quasi-periodic  operators  has    seen  significant progress through earlier  perturbative methods ~\cite{ds75,eli,fsw,fs,si87,cs89,eli97}, and then
	non-perturbative methods   ~\cite{j94,j,bj,bgreen05,bg20,hy12,afk11,gs08,gy20}. 
	We refer readers to
	~\cite{jmarx,you,jlz20,jcdm,sch1} and
	references therein for more details.   
	
	In this paper, we are interested in the   regime of large   coupling constants.  In this regime,   small denominator problems essentially become  problems of dealing with resonances coming from phases and frequencies ~\cite{bgs,bg20,fsw,j,jl18,jl2,bj,gs08}. 
	By avoiding the resonance (treating the small denominator)  that usually is achieved by imposing arithmetic conditions on phases and frequencies, it is expected that the operator exhibits Anderson localization (pure point spectrum with exponentially decaying eigenfunctions). 
	Recently,
  several remarkable sharp arithmetic transitions between singular continuous spectrum and pure point spectrum were obtained ~\cite{jl18,jl2,jlmary}. 
	However,  all localization proofs   ~\cite{j94,aj10,jl18,jl2,lyjfa,lyjst,liuetds20,j,bj,bgreen05,bg20,aj09,jy} simply bound 
	the numerators of  Green's functions of  the quasiperiodic Schr\"odinger operators   above  by the Lyapunov exponent of 
	the corresponding dynamical system \eqref{g122}.

The main goal of this paper is to initiate   an approach to   simultaneously treat   numerators  and   denominators  in the study of resonances arising from analytic quasi-periodic Schr\"odinger operators. As a starting point, we   focus on a particular case, namely, completely resonant phases for the almost Mathieu operator,  
in which   phase resonances and frequency resonances    overlap. As we were finalizing this paper we learned of a preprint \cite{HJY}  where the study of the numerators also plays an important  role, but in a completely different setting of unbounded potentials and  the so-called anti-resonances. 

%	The main goal of this paper is to initiate   an approach to   simultaneously treat   numerators  and   denominators  arising from quasi-periodic Schr\"odinger operators. As a starting point, we   focus on a particular case, namely, completely resonant phases for the almost Mathieu operator,  
%	in which   phase resonances and frequency resonances    overlap.

	%	Such upper bounds are  sufficient   in many cases ~\cite{jl18,jl2}. 
	%Such upper bounds are  sufficient  when only frequency resonances or phase resonances present ~\cite{jl18,jl2}. 
	%As we will see later, when frequency  and phase resonances appear at the same time,  the numerators are often smaller than the Lyapunov exponent.  %(referred to as Lyapunov-Perron non-regularity). 

	We say a  phase $\theta\in \mathbb{R}$ is completely resonant with respect to a  frequency $\alpha$ if $2\theta\in \alpha \mathbb{Z}+\mathbb{Z}$. One of the motivations   to study  completely resonant phases   is their critical roles in the spectral theory of quasiperiodic Schr\"odinger operators, for instance, spectral gap edges ~\cite{jfoot,puig04}. Another motivation comes from a conjecture of 
	Avila and Jitomirskaya in  ~\cite{aj09}  and Jitomirskaya ~\cite{ji94}.
	
	{\bf Conjecture 1:} Avila and Jitomirskaya  ~\cite{aj09} conjectured   that the almost Mathieu operator  $H_{\lambda,\alpha,\theta}$ satisfies Anderson localization  if  $ \ln |\lambda| >  2 \beta(\alpha)$ and  $2\theta\in \alpha \mathbb{Z}+\mathbb{Z}$,
	where
	\begin{equation}\label{gresf}
		\beta(\alpha)=\limsup_{k\rightarrow \infty}-\frac{\ln ||k\alpha||_{\R/\Z}}{|k|},
	\end{equation}
	and  $ ||x||_{\R/\Z}={\rm dist} (x,\Z) $.
	
%	Clearly, 
%		\begin{equation*}
% 	\beta(\alpha)=\limsup_{n\rightarrow\infty}\frac{\ln q_{n+1}}{q_n},
%	\end{equation*}
%	where   $ \frac{p_n}{q_n} $  is  the continued fraction  approximations    to $\alpha$.

%	{\bf Conjecture 2:} Avila and Jitomirskaya  ~\cite{aj09} conjectured   
%	
%		\begin{itemize}
%		\item  the almost Mathieu operator  $H_{\lambda,\alpha,\theta}$ satisfies Anderson localization  if  $  |\lambda| > e^{\beta(\alpha)+\delta(\alpha,\theta)}$;
%		\item the almost Mathieu operator  $H_{\lambda,\alpha,\theta}$  has purely singular continuous spectrum if  $ 1 <|\lambda| < e^{\beta(\alpha)+\delta(\alpha,\theta)}$.
%	\end{itemize}
%	
%	
%	that the almost Mathieu operator  $H_{\lambda,\alpha,\theta}$ satisfies Anderson localization  if  $ \ln |\lambda| >  2 \beta(\alpha)$ and  $2\theta\in \alpha \mathbb{Z}+\mathbb{Z}$,
%where
%\begin{equation}\label{gresf}
%	\beta(\alpha)=\limsup_{k\rightarrow \infty}-\frac{\ln ||k\alpha||_{\R/\Z}}{|k|},
%\end{equation}
%and  $ ||x||_{\R/\Z}={\rm dist} (x,\Z) $.

	Conjecture 1  is also  a particular case of a conjecture  of Jitomirskaya  ~\cite{ji94}.
	In this paper, we prove Conjecture 1 as it is.
	\begin{theorem}\label{mainthm}
		Assume that   $\alpha$  satisfies $\beta (\alpha)<\infty$.  Then the   almost Mathieu operator
		$H_{\lambda,\alpha,\theta}$ satisfies Anderson localization  if $2\theta\in \alpha \mathbb{Z}+\mathbb{Z}$  and  $ \ln |\lambda| > 2\beta(\alpha)$. 
		Moreover, if $\phi$ is an eigenfunction, that is $H_{\lambda,\alpha,
			\theta}\phi=E\phi$ with $\phi\in\ell^2(\Z)$, we have
		\begin{equation*}
			\limsup_{k\to \infty} \frac{\ln (\phi^2(k)+\phi^2(k-1))}{2|k|} \leq -(\ln |\lambda|-2\beta(\alpha)).
		\end{equation*}
	\end{theorem}
	Before  going  into the historical results attempting to solve Conjecture 1,  we   introduce  the phase and frequency resonances in the supercritical regime (this is also called the positive Lyapunov exponent regime).
	Define
	\begin{equation}\label{gres}
		\delta(\alpha,\theta)=\limsup_{k\to \infty}-\frac{\ln ||2\theta+k\alpha||_{\R/\Z}}{|k|}.
	\end{equation}
	
	When $\lambda \to \infty$, after a rescaling,   $H_{\lambda,\alpha,\theta}$  is  a perturbation of a diagonal matrix that has well localized eigenfunctions. 
	The diagonal entries (also  eigenvalues)  $\{2\lambda\cos(2\pi (\theta+k\alpha))\}_{k\in\Z}$ are  dense and the distance  between  any two entries  $2\lambda\cos(2\pi (\theta+k_1\alpha)) $ and  $2\lambda\cos(2\pi (\theta+k_2\alpha)) $ is 
	$$	|	2\lambda\cos(2\pi (\theta+k_1\alpha)) -2\lambda\cos(2\pi (\theta+k_2\alpha))| \;\;\;\;\;\;\; \;\;\;\;\;\;\; \;\;\;\;\;\;\; \;\;\;\;\;\;\; \;\;\;\;\;\;\;$$
	\begin{align*}
\;\;\;\;\;\;\;\;\;\;&=4|\lambda\sin (\pi (2\theta+ (k_1+k_2) \alpha))\sin(\pi(k_1-k_2)\alpha)|\\
	&\simeq |\lambda|\;\; || 2\theta+ (k_1+k_2) \alpha||_{\R/\Z} || (k_1-k_2)\alpha||_{\R/\Z}.
	\end{align*}
%	\begin{equation*}
	%	|	2\lambda\cos(2\pi (\theta+k_1\alpha)) -2\lambda\cos(2\pi (\theta+k_2\alpha))| =4|\lambda\sin (\pi (2\theta+ (k_1+k_2) \alpha))\sin(\pi(k_1-k_2)\alpha)|.
%	\end{equation*}
	By approximating  to the infinite  dimensional matrix $H_{\lambda,\alpha,\theta}$ from  its finite dimensional cut off, 
	the resonances $||2\theta+k\alpha||_{\R/\Z}$ (referred to as  phase resonances) and   $||k\alpha||_{\R/\Z}$ (referred to as frequency resonances)  appear. The strength of  frequency and phase resonances are essentially  quantified by \eqref{gresf} and \eqref{gres}.

	With recent developments, the   almost Mathieu operator undergoes a clear spectral transition (called metal-insulator transition or first transition line) when $|\lambda| $ changes from small to large.
	
	{\bf First transition line:}
	\begin{itemize}
		\item if $|\lambda|<1$,  $H_{\lambda,\alpha,\theta} $ has purely absolutely continuous spectrum for every $\alpha$ and $\theta $ ~\cite{j,ad08,aj10,avi08,gjls}.
		\item if $|\lambda|=1$,  $H_{\lambda,\alpha,\theta} $ has purely singular continuous spectrum for any irrational $\alpha$ and $\theta$ ~\cite{ajm,j19}.
		\item if $|\lambda|>1$, $H_{\lambda,\alpha,\theta} $ has  Anderson localization  for $\beta(\alpha)=\delta(\alpha,\theta)=0$ ~\cite{jks05,j}.
	\end{itemize}
	Understanding the first transition line  $|\lambda|=1$   is  now clear  by Avila's  global theory ~\cite{a15}.  When  $|\lambda|<1$ (subcritical regime) the Lyapunov exponent on the spectrum is $0$ even after the complexification of the  phase.  When $|\lambda|= 1$ (critical regime),  the Lyapunov exponent on the spectrum is $0$  and it changes to positive after the complexification. When $|\lambda|>1$ (supercritical regime),  the Lyapunov exponent is positive and equals $\ln|\lambda|$ on the spectrum.

	The spectral theory in the supercritical regime is more delicate.
	By Kotani theory, $H$ does not have any absolutely continuous spectrum  in the supercritical regime ~\cite{kot,ls99}. 
	 It is   believed that the spectral type of $ H$   (singular continuous spectrum or pure point spectrum/localization) depends on the competition between the Lyapunov exponent and resonances associated with  \eqref{gresf} and \eqref{gres}.  Such  observations go back to ~\cite{js94,j94,gordon,as82,j}.
	
	{\bf Second transition line in the frequency:}
	
	\begin{itemize}
		\item if  $ \ln |\lambda|>\beta(\alpha)$,  $H_{\lambda,\alpha,\theta} $ has   Anderson localization  for  $\delta(\alpha,\theta)=0$
		~\cite{jl18}.
		\item  if $0<\ln |\lambda|<\beta(\alpha)$,   $H_{\lambda,\alpha,\theta} $ has purely singular continuous spectrum ~\cite{ayz}.
	\end{itemize}
	
	{\bf Second  transition line in the phase:}
	\begin{itemize}
		\item if  $\ln |\lambda|>\delta(\alpha,\theta)$,  $H_{\lambda,\alpha,\theta} $ has   Anderson localization  for  $\beta(\alpha)=0$
		~\cite{jl2}.
		\item  if  $\ln |\lambda|<\delta(\alpha,\theta)$,   $H_{\lambda,\alpha,\theta} $ has purely singular continuous spectrum ~\cite{jl2}.
	\end{itemize}
	The second transition line  states that when the Lyapunov exponent beats frequency/phase resonances,  the operator exhibits localization. 
%	The second transition line in both frequency and phase solves the conjecture of Jitomirskaya  in 1994 ~\cite{ji94}.
	In the regime of localization, Jitomirskaya and the author   also determined  exact exponential asymptotics of eigenfunctions and corresponding transfer matrices. Moreover, 
	the (reflective) hierarchical structure  of eigenfunctions was discovered  in ~\cite{jl18,jl2}.

%	By reducibility and Aubry duality, Avila, You and Zhou  ~\cite{ayz}  (see an alternative proof by Jitomirskaya and Kachkovskiy ~\cite{jkmrl}) provided a solution of the measure-theoretic version of the second transition line in frequency. However, there is no arithmetic control on phases in ~\cite{ayz}.  
%	Recently, Ge and You ~\cite{gy20} developed the arithmetical version of the  Aubry duality,  which allows them to obtain the  localization  in the supercritical regime  with  arithmetical estimates on phases. In particular,   Ge and You offered a new proof to the results of Jitomirskaya ~\cite{j}. 
	
	Let us turn back to   completely  resonant phases. 	For   completely resonant phases, 
%	$2\theta\in \alpha\mathbb{Z}+\mathbb{Z}$, 
	$\delta(\alpha,\theta)=\beta(\alpha)$.
Thus  both phase resonances and frequency resonances  appear and  have the same strength. 
 Conjecture 1 states that when the Lyapunov exponent beats  phase and frequency resonances, the operator has Anderson localization.  It is natural to think that the proof of Conjecture 1 could follow from a combination of the second transition line in the phase and the frequency.  However, the completely resonant phase brings many new challenges, in particular the completely resonant phenomenon. 
The original arguments of Jitomirskaya ~\cite{j}  do not work  for completely  resonant phases directly.
In ~\cite{jks05}, Jitomirskaya-Koslover-Schulteis found a trick to  fix the gap by
shrinking the size of  intervals around 0 (referred  to  as the ``shrinking scale'' technique). This  allows them to prove 
Conjecture 1 when $\beta(\alpha)=0$.  The proof in ~\cite{aj10,lyjfg} implies that $H_{\lambda,\alpha,\theta}$ has Anderson localization if $\ln |\lambda|>C\beta(\alpha)$, where $C$ is an absolute constant.
In ~\cite{aj09}, Avila and Jitomirskaya  proved that when $\delta(\alpha,\theta)=0$ and $\ln |\lambda|>\frac{16}{9}\beta(\alpha)$, $H_{\lambda,\alpha,\theta}$ has Anderson localization, which is a key step to solve the ten Martini problem.  This approach has been pushed to the limit $\ln |\lambda|>\frac{3}{2}\beta(\alpha)$ by the author and Yuan ~\cite{lyjst}. 
By a combination of arguments in  Jitomirskaya-Koslover-Schulteis ~\cite{jks05}, Avila-Jitomirskaya ~\cite{aj09} and Liu-Yuan  ~\cite{lyjst}, the author and Yuan 
established the Anderson localization for completely resonant phases when $\ln |\lambda|>7\beta(\alpha)$ ~\cite{lyjfa}.

With  localization proofs in ~\cite{jl18,jl2} and methods in ~\cite{jks05,lyjfa}, we could obtain the  Anderson localization for  $\ln |\lambda|>4\beta(\alpha)$ in Conjecture 1, where  $4$ is the non-trivial technical limit in such an approach since we have  to shrink the scale to avoid the complete resonance,  doubling the numerical number.
In 
~\cite{liuetds20}, the author noticed that for completely resonant phases, the phase resonance and frequency resonance are not   symmetric, namely, the frequency resonance only happens at sites $k=jq_n$, $j\in\Z\backslash\{0\}$ and the phase resonance happens at both $k=jq_n+\frac{1}{2}q_n $ and $k=jq_n$, $j\in\Z$ ( $ \frac{p_n}{q_n} $  is  the continued fraction  approximations    to $\alpha$). See 	Sections \ref{S4} and \ref{S5} for more details.  	So, instead of using    the Lagrange interpolation uniformly, the author treated  the Lagrange interpolation individually during the process of finding  Green's functions  without  ``small denominators'',  which allows him to
push the numerical number $4$ to $3$ ~\cite{liuetds20}. % The numerical number $3$ is the technical limit. 
%It remains challenging to  the final solution of  Conjecture 1.

As we mentioned earlier, all  localization proofs  in previous literature ~\cite{aj09,aj10,lyjfa,lyjst,liuetds20,jks05} are devoted   to treating  small denominators (establishing  the lower bound of the denominators of Green's functions) and simply bound   numerators
by the rate of the Lyapunov exponent. 
%Our new observation  is  that resonances could  impact the numerators. 
 To the best of our knowledge, this  paper  is the first time to analyze numerators  of  Green's functions arising from quasi-periodic operators.  
 Moreover, we    expand     approaches in ~\cite{aj09,aj10,lyjfa,lyjst,liuetds20,jks05} to deal with denominators in many directions and add several significant ingredients to treat the numerators and denominators simultaneously.
 
 For completely resonant phases, the most challenging case  comes from resonant  sites $k=jq_n$, $j\in\Z\backslash\{0\}$ and   $k=jq_n+\frac{1}{2}q_n $, $j\in\Z$. 
 The Gordon type argument is an approach  to show the absence of eigenvalues based on repetitions of potentials governed by frequency resonances ~\cite{gordon,ayz,jl18,jlmary,as82}.
 We   develop  the Gordon type argument to establish sharp bounds of the numerators of  Green's functions around resonant sites $jq_n$. However, such an argument  does not work  for resonant sites $jq_n+\frac{1}{2} q_n$ (only coming from phase resonances) since the  repetition of potentials only appears at sites $jq_n$ by the   continued fraction expansion approximation. We regard  norms of generalized eigenfunctions at resonant sites  $k=jq_n, q_n+\frac{1}{2}q_n $ and numerators of  Green's functions around  resonant sites $jq_n+\frac{1}{2}q_n$ as variables. 
 Our first step is to establish inequalities among those variables.  By solving those inequalities, we obtain the relations among resonant sites. Then the exponential  decay at the resonant sites follows from  standard iterations.
 
Finally, we want to highlight another technical   novelty in our proof. The Gordon type argument works along with the transfer matrices,  which slightly differs from the denominators  and numerators of  Green's functions.  
 We introduce  the  notation $P_{[x_1,x_2]}$
to represent the denominator of  the Green's function restricting  to the interval $[x_1,x_2]$.  The advantage  of   the notation  $P_{[x_1,x_2]}$ is that it 
inherits  information from  the  transfer matrix from sites $x_1$ to $x_2$. This could simplify   the localization proof  in \cite{aj09,aj10,lyjfa,lyjst,liuetds20,jks05}.

	The rest of this paper is organized as follows. We list the definitions and standard facts  in Section \ref{S2}. 
	In Section \ref{S3}, we provide several technical lemmas.
	Sections \ref{S4} and \ref{S5} are devoted to treating  resonances. 
	In Section \ref{S6}, we complete the proof.

	\section{ Some notations and   known facts  }\label{S2}

	Let $H$ be an operator on $\ell^2(\Z)$.
	We say $\phi$ is a generalized eigenfunction corresponding to the generalized eigenvalue $E$ if 
	\begin{equation}\label{g68}
		H\phi=E\phi  ,\text{ and }  |\phi(k)|\leq \hat{C}( 1+|k|).
	\end{equation}
	By Shnol's theorem  ~\cite{berezanskii1968expansions},  in order to
	prove  Anderson localization of   $H$,  we only need  to show that  every  generalized eigenfunction  is in fact an exponentially decaying eigenfunction.
	To be more precise,  there exists some constant $c>0$ such that
	\begin{equation*}%\label{G22}
		| \phi(k)|\leq  e^{-c|k|} \text{ for large } k.
	\end{equation*}	
	For simplicity, we assume    $\hat{C}=1$ in \eqref{g68}.
	
	%	It suffices to consider $ \alpha$ with  $ 0<\beta(\alpha) <\infty$.  

	From now on, we always assume $\phi$  is   a  generalized  eigenfunction  of $H_{\lambda,\alpha,\theta}$ and $E$ is the corresponding generalized eigenvalue.  Without loss of generality assume $\phi(0)=1$. It is well known that every generalized eigenvalue must be in the spectrum, namely $E\in \Sigma_{\lambda,\alpha}$, where
	$\Sigma_{\lambda,\alpha}$ is  the spectrum of   $H_{\lambda,\alpha,\theta}$ (the spectrum  does not depend on $\theta$).   
	Our  goal  is to show that there exists some  constant $c>0$  such that for large $k$,
	$$| \phi(k)|\leq  e^{-c|k|} .$$
	For any $x_1,x_2\in\Z$ with $x_1<x_2$, denote by 
	\begin{equation*}
		P_{[x_1,x_2]}(\lambda,\alpha,\theta,E)	=\det(R_{[x_1,x_2]}(H_{\lambda,\alpha,\theta}-E) R_{[x_1,x_2]}),
	\end{equation*}	
where $R_{[x_1,x_2]}$ is the restriction on $[x_1,x_2]$. 
	Let us denote
	$$ P_k(\lambda,\alpha,\theta,E)=\det(R_{[0,k-1]}(H_{\lambda,\alpha,\theta}-E) R_{[0,k-1]}).$$
	When there is no ambiguity, we drop the  dependence of  parameters $E,\lambda,\alpha$ or $ \theta$.
	Clearly,
	\begin{equation}\label{g108}
		P_{[x_1,x_2]}(\theta)	=P_{k} (\theta+x_1\alpha),
	\end{equation}
	where $k=x_2-x_1+1$.
	%Following $ {~\cite{J}}$,
	%	It is easy to see that $P_k(\theta)$ is an even function of $ \theta+\frac{1}{2}(k-1)\alpha$  and can be written as a polynomial
	%	of degree $k$ in $\cos2\pi (\theta+\frac{1}{2}(k-1)\alpha )$ :
	
	%	\begin{equation}\label{GP_k}
	%		P_k(\theta)=\sum _{j=0}^{k}c_j\cos^j2\pi (\theta+\frac{1}{2}(k-1)\alpha)    \triangleq  Q_k(\cos2\pi  (\theta+\frac{1}{2}(k-1)\alpha)).
	%	\end{equation}
	
	Let
	\begin{equation}\label{G.transfer}
		A_{k}(\theta)=\prod_{j=k-1}^{0 }A(\theta+j\alpha)=A(\theta+(k-1)\alpha)A(\theta+(k-2)\alpha)\cdots A(\theta)
	\end{equation}
	and
	\begin{equation}\label{G.transfer1}
		A_{-k}(\theta)=A_{k}^{-1}(\theta-k\alpha)
	\end{equation}
	for $k\geq 1$,
	where $A(\theta)=\left(
	\begin{array}{cc}
		E-2\lambda\cos2\pi\theta & -1 \\
		1& 0\\
	\end{array}
	\right)
	$.
	$A_{k}$  is called the (k-step) transfer matrix. 
	
	By the definition,   for  any $k\in\Z_+, m\in\Z$, one has
	\begin{equation}\label{G.new17}
		\left(\begin{array}{c}
			\phi(k+m) \\
			\phi(k+m-1)                                                                                        \end{array}\right)
		=A_{k}(\theta+m\alpha)
		\left(\begin{array}{c}
			\phi(m) \\
			\phi(m-1)                                                                                        \end{array}\right).
	\end{equation}
	
	It is easy to check  that for $k\in\Z_+$, 
	\begin{equation}\label{G34}
		A_{k}(\theta)=
		\left(
		\begin{array}{cc}
			P_k(\theta) &- P_{k-1}(\theta+\alpha)\\
			P_{k-1}(\theta) & - P_{k-2}(\theta+\alpha) \\
		\end{array}
		\right).
	\end{equation}
	The Lyapunov exponent %for  transfer matrix  $A$ (notice that $A$
	%depends on $E$ )
	is given  by
	\begin{equation}\label{G21}
		L(E)=\lim_{k\rightarrow\infty} \frac{1}{k}\int_{\mathbb{R}/\mathbb{Z}} \ln \| A_k(\theta)\|d\theta.
	\end{equation}
	The Lyapunov exponent can be computed precisely for $E$ in the
	spectrum of $H_{\lambda,\alpha,\theta}$. 
	\begin{lemma}~\cite{bj02}\label{lya}
		For $E\in \Sigma_{\lambda,\alpha}$ and $|\lambda|>1$, we have
		$L(E)=\ln|\lambda|$.
	\end{lemma}
	In the following,  denote by 
	$L := \ln|\lambda| $.
	%be the Lyapunov exponent of the almost Mathieu operator for energies in the spectrum.
	
	By
	upper semicontinuity and unique ergodicity,
	%By Corollary 2   in   ~\cite{furman1997multiplicative}  (since irrational rotations are uniquely ergodic),
	one has
	\begin{equation}\label{G23}
		L=\lim_{k\rightarrow\infty} \sup_{\theta\in\mathbb{R}/ \mathbb{Z}}\frac{1}{k} \ln \| A_k(\theta)\|.
	\end{equation}
	Therefore, 
	for any  $  \varepsilon >0$,  
	\begin{equation}\label{G24}
		\| A_k(\theta)\|\leq  e^{(L+\varepsilon)k},
	\end{equation}
	when  $k$ is large enough  (independent of $\theta$).
	
	By \eqref{G34} and \eqref{G24},  one has for large $k$,
	\begin{equation}\label{Numerator}
		| P_k(\theta)|\leq  e^{(L+\varepsilon)k},
	\end{equation}
	and hence
	\begin{equation}\label{Numerator1}
		| P_{[x_1,x_2]}(\theta)|\leq  e^{(L+\varepsilon)|x_2-x_1|}.
	\end{equation}
	By \eqref{G.new17} and \eqref{G24}, one has that  for large $|k_1-k_2|$,
	\begin{equation}\label{g500}
		\left \|	\left(\begin{array}{c}
			\phi(k_1+1) \\
			\phi(k_1)                                                                                        \end{array}\right)\right\|
		\leq  e^{(L+\varepsilon) |k_1-k_2|}\left\|
		\left(\begin{array}{c}
			\phi(k_2+1) \\
			\phi(k_2)                                                                                        \end{array}\right)\right\|.
	\end{equation}
	For any $x_1,x_2\in\Z$ with $x_1<x_2$,  let $G_{[x_1,x_2]}$  be the Green's function:
	\begin{equation*}
	G_{[x_1,x_2]}=(R_{[x_1,x_2]} (H_{\lambda,\alpha,\theta}-E) R_{[x_1,x_2]})^{-1}.
	\end{equation*}
	By Cramer's rule (see p.15, ~\cite{bgreen05} for example)  for  any
	$ y\in [x_1,x_2]  $,  one has
	\begin{eqnarray}
		|G_{[x_1,x_2]}(x_1,y)| &=&  \left| \frac{P_{[y+1,x_2]}}{P_{[x_1,x_2]}}\right|,\label{Cramer1}\\
		|G_{[x_1,x_2]}(y,x_2)| &=&\left|\frac{P_{[x_1,y-1]}}{P_{[x_1,x_2]}} \right|.\label{Cramer2}
	\end{eqnarray}
	It is  easy to check that (p. 61, ~\cite{bgreen05}) for  any 	$ y\in [x_1,x_2]  $, 
	\begin{equation}\label{Block}
		\phi(y)= -G_{[x_1,x_2]}(x_1,y ) \phi(x_1-1)-G_{[x_1,x_2]}(y,x_2) \phi(x_2+1).
	\end{equation}

	Denote by   $x_1^\prime=x_1-1$ and $x_2^\prime=x_2+1$.
	
	By \eqref{Cramer1}, \eqref{Cramer2} and \eqref{Block}, 
	one has that  for any $y\in [x_1,x_2]$, 
	\begin{equation}\label{g100}
		|\phi(y)|\leq |P_{[x_1,x_2]}|^{-1} |P_{[y+1,x_2]}|  |\phi(x_1^{\prime})| + |P_{[x_1,x_2]}|^{-1}  |P_{[x_1,y-1]}| |\phi(x_2^{\prime})|.
	\end{equation}

	Given  a set $\{\theta_1, \cdots ,\theta_{k+1}\}$,  the lagrange Interpolation terms $\text{Lag}_m$, $m=1,2,\cdots,k+1$, are  defined by
	\begin{equation}\label{Def.Uniform}
		\text{Lag}_m= \ln \max_{ x\in[-1,1]} \prod_{ j=1 , j\neq m }^{k+1}\frac{|x-\cos2\pi\theta_j|}
		{|\cos2\pi\theta_m-\cos2\pi\theta_j|}.
	\end{equation}
	
	% Let $A_{k,r}=\{\theta\in\mathbb{R} \;|\;P_k(\cos2\pi  ( \theta -\frac{1}{2}(k-1)\alpha )  )|\leq e^{(k+1)r}\} $ with $k\in \mathbb{N}$ and $r>0$,
	%we have the following lemma.
	The following lemma is another form of  Lemma 9.3 in ~\cite{aj09}, which has been reformulated  in  ~\cite[Lemma 2.3]{liuetds20}.
	\begin{lemma}~\cite[Lemma 2.3]{liuetds20}\label{Le.Uniform}
		Given a set $\{\theta_1, \cdots ,\theta_{k+1}\}$,   there exists some $\theta_m$ in  $\{\theta_1, \cdots ,\theta_{k+1}\}$ such that
		
		\begin{equation*}
			|	P_{k}(\theta_m -\frac{k-1}{2}\alpha)|\geq \frac{e^{k L-\rm{Lag}_m}}{k+1}.
		\end{equation*}
		
	\end{lemma}

	In the following,  we always assume 
	\begin{itemize}
		\item $\varepsilon>0$ is an arbitrarily small constant and it may change even in the same equation.
		\item  $C$ is a large constant (depends on $\lambda$ and $\alpha$) and  it may change even in the same equation.
		\item $n$ is large enough  which depends on all  parameters and constants.
	\end{itemize}
	%The following lemma   Theorem 2.4, which is similar to the proof of ~\cite[Lemma 4.1]{aj09}  and ~\cite[Lemma 2.5]{liuetds20}
Recall that	 $ \frac{p_n}{q_n} $  is  the continued fraction  approximations    to $\alpha$.
	By the definition of $\beta(\alpha)$, one has 
that 
\begin{equation*}
\beta=	\beta(\alpha)=\limsup_{n\rightarrow\infty}\frac{\ln q_{n+1}}{q_n}.
\end{equation*}
	
%		By the definition of $\beta(\alpha)$, one has 
%	that 
%	\begin{equation*}
%\beta=	\beta(\alpha)=\limsup_{n\rightarrow\infty}\frac{\ln q_{n+1}}{q_n},
%	\end{equation*}
%	where   $ \frac{p_n}{q_n} $  is  the continued fraction  approximations    to $\alpha$.
%	
	Let $b_n=   10^{-5}q_{n}$.	
	For any $\ell \in\Z$, let
	\begin{equation*}
		r_{\ell}^{\varepsilon,n}= \sup_{|r|\leq 10\varepsilon}|\phi(\ell q_n+rq_n)|,
	\end{equation*}
	and
	\begin{equation*}
		r_{\ell+\frac{1}{2}}^{\varepsilon,n}= \sup_{|r|\leq 10\varepsilon}|\phi(\ell q_n+\lfloor\frac{q_n}{2}\rfloor+rq_n)|,
	\end{equation*}
	where  $\lfloor x \rfloor$ is the largest integer that is less than or equal to $x$.
	\begin{lemma}\label{Le.resonant}
		~\cite[Lemma 2.5]{liuetds20}
		Assume that $|\lambda|>1$ and $2\theta\in\alpha\Z+\Z$.
		Suppose $k\in [\ell q_n,(\ell+\frac{1}{2})q_n]$  or $k\in [(\ell+\frac{1}{2})q_n,(\ell+1)q_n]$  with $0\leq |\ell| \leq  100\frac{b_{n+1}}{q_n}+100$,  and  $  {\rm dist}(k, q_n\mathbb{Z}+\frac{q_n}{2}\mathbb{Z} )\geq     10 \varepsilon q_n$. Let  $d_t=|k-tq_n|$ for $t\in\{\ell,\ell+\frac{1}{2},\ell+1\}$.
		Then for sufficiently large
		$n $, we have  that 
		\begin{itemize}
			\item 	when $k\in [\ell q_n,(\ell+\frac{1}{2})q_n]$, 
			\begin{equation}\label{Intervalk}
				|\phi(k)|\leq r_{\ell}^{\varepsilon,n}\exp\{-(L- \varepsilon)(d_{\ell}-3\varepsilon q_n)\} + r_{\ell+\frac{1}{2}}^{\varepsilon,n}\exp\{-(L- \varepsilon)(d_{\ell+\frac{1}{2}}-3\varepsilon q_n)\};
			\end{equation}
			\item 
			when $k\in [(\ell+\frac{1}{2})q_n,(\ell+1)q_n]$,
			\begin{equation}\label{Intervalk1}
				|\phi(k)|\leq r_{\ell+\frac{1}{2} }^{\varepsilon,n}\exp\{-(L- \varepsilon)(d_{\ell+\frac{1}{2}}-3\varepsilon q_n)\}+r_{\ell+1}^{\varepsilon,n}\exp\{-(L- \varepsilon)(d_{\ell+1}-3\varepsilon q_n)\}.
			\end{equation}
		\end{itemize}

	\end{lemma}
	A similar version of Lemma \ref{Le.resonant} firstly appears in ~\cite[Lemma 4.1]{jl18}.  
	%The formulation of Lemma \ref{Le.resonant} was given in ~\cite{liuetds20}.
	\section{Technical preparations}\label{S3}

	Without loss of generality, we assume $\alpha\in(0,1)$ and $\lambda>e^{2\beta(\alpha)}$.  Since $2\theta\in\alpha \Z+\Z$,   essentially we only need to study   $\theta\in\{ -\frac{\alpha}{2},-\frac{\alpha}{2}-\frac{1}{2},0, -\frac{1}{2}\}$ by   shifting the operator from $H_{\lambda,\alpha,\theta}$ 
	to $H_{\lambda,  \alpha,\theta\pm\alpha}$. For this reason, we always assume $\theta\in\{ -\frac{\alpha}{2},-\frac{\alpha}{2}-\frac{1}{2},0, -\frac{1}{2}\}$ in the following arguments.
	
%	Denote by $\lfloor x\rfloor$ the  largest integer less or equal  than  $x$. 
%	For any $j\in\Z$, let
%	\begin{equation*}
%		r_j= \sup_{|r|\leq 10\varepsilon }|\phi(jq_n+rq_n)|,
%	\end{equation*}
%	and
%	\begin{equation*}
%		r_{j+\frac{1}{2}}= \sup_{|r|\leq 10\varepsilon }|\phi(jq_n+\lfloor\frac{q_n}{2}\rfloor+rq_n)|.
%	\end{equation*}
%	As is clear from the definition, 
	
For simplicity, we drop superscripts $n$ and $\varepsilon$ from $ 	r_j^{\varepsilon,n}$ and $	r_{j+\frac{1}{2}}^{\varepsilon,n}$,
since  $n$ and $\varepsilon$ will be   fixed.
%$ 	r_j$ and $	r_{j+\frac{1}{2}}$  also depend on $n$ and $\varepsilon$.  Since  $n$ and $\varepsilon$ will be   fixed, we omit them from the notation.
	\begin{lemma}\label{klem2}
		Let $0\leq x\leq \frac{q_n}{4}$.  Assume that $p_1$ satisfies $|p_1-\frac{1}{2} q_n|\leq 20 \varepsilon q_n$ and $p_2$ satisfies 
		$|p_2- q_n|\leq 20 \varepsilon q_n$. 
		Then we have  
		\begin{equation}\label{g10}
			|	P_{ [-x,p_1]}  |\leq e^{L(\frac{1}{2}q_n-x)+C\varepsilon q_n},
		\end{equation}
		\begin{equation}\label{g101}
			|	P_{ [-p_1,x]}  |\leq e^{L(\frac{1}{2}q_n-x)+C\varepsilon q_n},
		\end{equation}
		and
		\begin{equation}\label{g102}
			|	P_{ [-x,p_2]} |\leq e^{L( q_n-x)+C\varepsilon q_n}.
		\end{equation}
	\end{lemma}
	\begin{proof}
		We are going to prove \eqref{g10} first.
		Let $I=[x_1,x_2]$, where $x_1=-x$ and $x_2=p_1$.  
		
		By Lemma \ref{Le.resonant}, one has
		\begin{equation}\label{g11}
			|\phi(x_1^\prime)\leq r_{0} e^{-L  x+C\varepsilon q_n} +r_{-\frac{1}{2}} e^{-L ( \frac{q_n}{2} -x)+C\varepsilon q_n} 
		\end{equation}
		%	Recall that 
		%	\begin{equation}\label{g12}
		%		|\phi(x_2^\prime)|\leq r_{\frac{1}{2}} .
		%	\end{equation}
		By   \eqref{g100}, one has
		\begin{equation}\label{g13}
			|\phi(0)| \leq 	|	P_{ [-x,p_1]} |^{-1}  (|P_{[-x,-1]}| |\phi (x_2^\prime)|+
			|P_{[1,p_1]} | |\phi (x_1^\prime)|).
		\end{equation}
		
		By \eqref{Numerator1}, \eqref{g11}  and \eqref{g13}, one has
		\begin{equation}\label{g14}
			|\phi(0)| \leq 	|	P_{ [-x,p_1]} |^{-1}  e^{C\varepsilon q_n}(e^{L(\frac{q_n}{2}-x)}r_0+e^{L x} r_{-\frac{1}{2}}+e^{L x} r_{\frac{1}{2}}).
		\end{equation}
		Since $\phi(0)=1$ and $|\phi(k)|\leq 1+|k|$, by \eqref{g14}, we have that
		\begin{equation*} 
			|	P_{ [-x,p_1]} |\leq e^{L(\frac{1}{2}q_n-x)+C\varepsilon q_n}.
		\end{equation*}
		
		We are going to prove \eqref{g101} (\eqref{g102}). In this case, we only need to 
		set  $I=[x_1,x_2]$, where $x_1=-p_1$ and $x_2=x$ ($x_1=-x$ and $x_2=p_2$), and repeat the proof of \eqref{g10}. 
		
	\end{proof}
	
	\begin{lemma}~\cite{simon1985almost}\label{Lemgordonidea2}
		Let $A^1,A^2,\cdots,A^n$ and  $B^1,B^2,\cdots,B^n$ be $2\times2$ matrices with $||\prod_{m=0}^{j-1}A^{k+m}||\leq D e^{dj}$ for some constant $D$ and $d$.
		Then
		\begin{equation*}
			||(A^n+B^n)\cdots(A^1+B^1)-A^n\cdots A^1||\leq De^{dn} (\prod_{j=1}^n(1+De^{-d}||B^j||)-1).
		\end{equation*}
		
	\end{lemma}
	\begin{lemma}\label{klem1}
		Assume that  $0<k\leq 10 q_n$  and  $0<j\leq C\frac{q_{n+1}}{q_n}+C$. Then for large enough $k$,  we have
		\begin{equation}
			||A_k(\theta)-A_k(\theta +jq_n\alpha)|| \leq e^{(L+\varepsilon)k }\frac{j}{q_{n+1}}.
		\end{equation}
		In particular,
		\begin{equation}
			|P_k(\theta)-P_k(\theta +jq_n\alpha)| \leq e^{(L+\varepsilon) k} \frac{j}{q_{n+1}}.
		\end{equation}
	\end{lemma}
	\begin{proof}
		By the Diophantine    approximation (see \eqref{GDC2} in the Appendix), we have
		\begin{equation*}
			||{q}_n\alpha||_{\mathbb{R}/\mathbb{Z}}\leq \frac{1}{q_{n+1}} ,
		\end{equation*}
		and hence
		\begin{equation*}
			||j{q}_n\alpha||_{\mathbb{R}/\mathbb{Z}}\leq \frac{j}{q_{n+1}} .
		\end{equation*}
		This implies
		\begin{equation*}
			||A(\theta+jq_n\alpha)-A(\theta)||\leq \frac{Cj}{q_{n+1}}.
		\end{equation*}
		%where $D $ is a constant (depends on $E$ and $\lambda$).
		Applying \eqref{G24} and Lemma \ref{Lemgordonidea2},
		one has
		\begin{equation}\label{G.new21}
			||A_{k}(\theta+ j{q}_n\alpha)-A_{ k}(\theta)||\leq  e^{(L+\varepsilon)k}((1+\frac{Cj}{q_{n+1}})^{k}-1).
		\end{equation}
		Using the fact $|e^y-1|\leq ye^y$ for $y>0$,
		we obtain
		\begin{eqnarray*}
			(1+\frac{Cj}{q_{n+1}})^{k}-1  &\leq& k(1+\frac{jC}{q_{n+1}})^{k}\ln(1+\frac{Cj}{q_{n+1}})\\
			&\leq&  C\frac{kj}{q_{n+1}}.
		\end{eqnarray*}
		Combining this with (\ref{G.new21}) completes the proof.
	\end{proof}
	\section{ Resonance I: sites $jq_n+\frac{q_n}{2},j\in\Z$}\label{S4}

	In this section, we  deal with  resonances arising from sites $jq_n+\frac{q_n}{2}$, $j\in\Z$.

	Denote by 
	\begin{equation}
		\beta_j=\frac{\ln q_{n+1} -\ln (|j|+1)}{q_n} .
	\end{equation}
	
	We are going to prove
	\begin{theorem}\label{thm1}
		%Suppose $b_{n+1}\geq \frac{q_n}{4}$.
		Let $   |j |\leq 2\frac{b_{n+1}}{q_n}+10$. Then  we have
		%Suppose  $\frac{b_n}{4}\leq |k|\leq b_{n+1}$ is resonant and $|k|\geq \frac{q_n}{2}$, then
		\begin{equation}
			r_{j+\frac{1}{2}}  \leq   \exp\{- \frac{1}{2}(L-2\beta_{j} -C\varepsilon)q_n\}(r_{j}+r_{j+1}).
		\end{equation}
		
	\end{theorem}
	\begin{proof}
		
		Take  $\phi(j q_n+\lfloor\frac{q_n}{2}\rfloor+rq_n)$ with $|r|\leq 10\varepsilon $ into consideration.   Denote by $p=j q_n+\lfloor\frac{q_n}{2}\rfloor+rq_n$.
		Without loss of generality assume $j\geq0$.
		Let $n_0$ be the least positive integer such that
		\begin{equation*}
			q_{n-n_0}\leq  \frac{ \varepsilon}{2}  (\frac{1}{8}-2\varepsilon) q_n.
		\end{equation*}
		Let $s$ be the
		largest positive integer such that $sq_{n-n_0}\leq (\frac{1}{8}-2\varepsilon) q_n $.
		By the fact $(s+1)q_{n-n_0}\geq (\frac{1}{8}-2\varepsilon)q_n $, one has 
		\begin{equation*}
			s\geq \frac{1}{\varepsilon} 
		\end{equation*}
		and
		\begin{equation}\label{sq}
			(\frac{1}{8}-3\varepsilon)q_n \leq sq_{n-n_0}\leq (\frac{1}{8}-2\varepsilon)q_n.
		\end{equation}
		
		%For simplicity, we assume $n_0=1$.
		Set $I_1, I_2\subset \mathbb{Z}$ as follows
		\begin{eqnarray*}
			I_1 &=& [-2s q_{n-n_0},-1 ], \\
			I_2 &=& [ j q_n+\lfloor\frac{q_n}{2}\rfloor -2sq_{n-n_0},j q_n+\lfloor\frac{q_n}{2}\rfloor+2sq_{n-n_0}-1 ],
		\end{eqnarray*}
		and let $\theta_m=\theta+m\alpha$ for $m\in I_1\cup I_2$. The set $\{\theta_m\}_{m\in I_1\cup I_2}$
		consists of $6sq_{n-n_0}$ elements. Let $k=6s q_{n-n_0}-1$.
		
		%Considering set $I_1\cup I_2$,
		By  modifying  the proof of ~\cite[Lemma 9.9]{aj09} and    ~\cite[Lemma 4.1]{lyjfa} (or Appendices in ~\cite{jl18} and ~\cite{liuetds20}), we can prove the claim (Claim 1):
	 for any   
		$  m\in  I_2$,
		\begin{equation}\label{g66}
			\rm{Lag}_m\leq  q_n (\beta_{j}+ \varepsilon) ,
		\end{equation}
		and for any
		$  m\in  I_1$,
		\begin{equation}\label{g661}
			\rm{Lag}_m\leq  q_n \varepsilon .
		\end{equation}
		For convenience, we include the proof in the Appendix.

		By Lemma \ref{Le.Uniform}, there exists some $j_0\in  I_2$ such that
		\begin{equation}\label{g109}
			|P_k(\theta_{j_0}-\frac{k-1}{2}\alpha) |\geq e^{kL-(\beta_{j} +\varepsilon)q_n}.
		\end{equation}
		or  there exists some $j_0\in  I_1$ such that
		\begin{equation}\label{g1091}
			|P_k(\theta_{j_0}-\frac{k-1}{2}\alpha) |\geq e^{kL-\varepsilon q_n}.
		\end{equation}
		
		Let 
		\begin{equation}\label{g4}
			I=[j_0-3sq_{n-n_0} +1,j_0+3sq_{n-n_0}-1]=[x_1,x_2]. 
		\end{equation}

		By \eqref{g108} and \eqref{g109},  one has that when $j_0\in I_2$,
		\begin{equation}\label{g110}
			|	P_{[x_1,x_2]}(\theta)| 	=	|P_k(\theta_{j_0}-\frac{k-1}{2}\alpha)| \geq e^{kL-(\beta_{j} +\varepsilon)q_n}.
		\end{equation}
		Recall that   $x_1^\prime=x_1-1$ and $x_2^\prime=x_2+1$.
		By   \eqref{g100}, \eqref{g110} and \eqref{sq}, one has that   when $j_0\in I_2$,
		\begin{align}
			|\phi(p)| \leq& e^{-kL+\beta_{j} q_n+\varepsilon q_n } (|P_{[x_1,p-1]}| |\phi (x_2^\prime)|+
			|P_{[p+1,x_2]} | |\phi (x_1^\prime)|)\nonumber\\
			\leq &e^{-\frac{3}{4}Lq_n+\beta_{j} q_n+C\varepsilon  q_n } (|P_{[x_1,p-1]}| |\phi (x_2^\prime)|+
			|P_{[p+1,x_2]} | |\phi (x_1^\prime)|).\label{g2}
		\end{align}	
		
		%	where $p=jq_n+\lfloor\frac{q_n}{2}\rfloor+rq_n$.
		
		Clearly,
		\begin{equation*}
			I_1\subset [-\frac{q_n}{4},0]
		\end{equation*}
		and
		\begin{equation*}
			I_2\subset [jq_n+\frac{q_n}{4} ,jq_n+\frac{3}{4} q_n].
		\end{equation*}
		
		{\bf Case 1: }  $j_0\in [jq_n+\frac{3}{8} q_n, jq_n+\frac{5}{8}q_n]\cap I_2$

		In this case, by \eqref{sq} and \eqref{g4}, one has  that
		\begin{equation}\label{g250}
			x_1\in [jq_n+6\varepsilon q_n,jq_n+\frac{1}{4}q_n+9\varepsilon q_n+1], 	x_2\in [jq_n+\frac{3}{4}q_n-9\varepsilon q_n-1,jq_n+q_n-6\varepsilon q_n].
		\end{equation}
		In order to make  the following arguments neat, 
		% we are not going to  make the difference between $[a,b]$ and
		%$[a+a^\prime, b+b^\prime]$ 
		%\footnote{ In other words, instead of \eqref{sq}, we think $sq_{n-n_0}=\frac{1}{8}q_n$.}, where $|a^\prime|\leq  50 \varepsilon q_n$ and $|b^\prime|\leq 50\varepsilon q_n$.  
		we are not going to  make the difference between $a\in\Z$ and $a^\prime\in\Z$ if  $|a-a^\prime|\leq 50\varepsilon q_n$.
	%	This will not change the conclusion 
		 %\footnote{ In other words, instead of \eqref{sq}, we think $sq_{n-n_0}=\frac{1}{8}q_n$.}.
		%\footnote{ In other words, instead of \eqref{sq}, we think $sq_{n-n_0}=\frac{1}{8}q_n$.}, where $|a^\prime|\leq  50 \varepsilon q_n$ and $|b^\prime|\leq 50\varepsilon q_n$.  
		For example, instead of using \eqref{g250}, we simply  write
		\begin{equation}
			x_1\in [jq_n,jq_n+\frac{1}{4}q_n], 	x_2\in [jq_n+\frac{3}{4}q_n,jq_n+q_n].
		\end{equation}

		\begin{center}
			\begin{tikzpicture}[thick, scale=1]
				\draw[-](0,0)--(4.75,0); 		\draw[-](7.25,0)--(12,0); 	
				\draw[red,-](1,0)--(3.5,0); 
				\draw[blue,-](7.25,0)--(4.75,0); 
				\draw[red,-](8.5,0)--(11,0); 
				\draw[-](6,0)--(6,-0.2);
				\node [below]  at (6,-0.2){$jq_n+\frac{q_n}{2}$};
				\draw[-](1,0)--(1,-0.2);
				\node [below]  at (1,-0.2){$jq_n$};
				\draw[-](11,0)--(11,-0.2);
				\node [below]  at (11,-0.2){$jq_n+q_n$};
				\draw[-](3.5,0)--(3.5,0.2);
				\node [above]  at (3.5,0.2){$\frac{q_n}{4}$};
				\draw[-](8.5,0)--(8.5,0.2);
				\node [above]  at (8.5,0.2){$\frac{3q_n}{4}$};
				
				\draw[-](4.75,0)--(4.75,0.2);
				\node [above]  at (4.75,0.2){$\frac{3q_n}{8}$};
				\draw[-](7.25,0)--(7.25,0.2);
				\node [above]  at (7.25,0.2){$\frac{5q_n}{8}$};
				\node [above]  at (6,0.5){$j_0$};
				\draw[->] (6,0.5)--(5.8,0);
				
				\node [above]  at (2,0.5){$x_1$};
				\draw[->] (2,0.5)--(1.8,0);
				\node [above]  at (10,0.5){$x_2$};
				\draw[->] (10,0.5)--(10.2,0);

				\node [below] at (6,-1){Fig.1};

			\end{tikzpicture}
		\end{center}
		In Fig.1, $x_1$ and $x_2$ locate at the red part. $j_0$ locates at the blue part. The  numbers above  are  the sites after deducting  $jq_n$.

		By \eqref{Numerator1}, one has 
		\begin{equation}\label{g1}
			|P_{[x_1,p-1]}|\leq e^{( L+\varepsilon)|p-x_1|},  
		\end{equation}
		and
		\begin{equation}\label{g0}
			|P_{[p+1,x_2]}|\leq e^{( L+\varepsilon)|x_2-p|}.
		\end{equation}
		By \eqref{g2}, \eqref{g1} and \eqref{g0},  one has that 
		\begin{equation}\label{Iterationr_j}
			|\phi(p)|  \leq \sum_{i=1,2} e^{( \beta_{j} +C\varepsilon)q_{n}}|\phi(x_i^\prime)|e^{-|p-x_i|L }.
		\end{equation}
		By Lemma \ref{Le.resonant}, one has
		\begin{equation}\label{g6}
			|\phi(x_1^\prime)\leq r_{j} e^{-L (x_1-jq_n)+C\varepsilon q_n} +r_{j+\frac{1}{2}} e^{-L (jq_n+\frac{q_n}{2} -x_1)+C\varepsilon q_n} 
		\end{equation}
		and 
		\begin{equation}\label{g5}
			|\phi(x_2^\prime)\leq r_{j+\frac{1}{2}} e^{-L (x_2-jq_n-\frac{q_n}{2})+C\varepsilon q_n} +r_{j+1}e^{-L (jq_n+q_n-x_2)+C\varepsilon q_n} .
		\end{equation}
		By \eqref{Iterationr_j}, \eqref{g6}, \eqref{g5}  and the fact that $|x_i-p|\geq \frac{q_n}{4}-C\varepsilon q_n$, $i=1,2$, 
		we have
		\begin{eqnarray}
			r_{j+\frac{1}{2}} &\leq&  e^{-\frac{1}{2}(L-2\beta_{j}- C\varepsilon)q_n} r_{j}+e^{-\frac{1}{2}(L-2\beta_{j}- C\varepsilon)q_n}r_{j+1} \nonumber\\
			& & + e^{-\frac{1}{2}(L-2\beta_{j}- C\varepsilon)q_n} r_{j+\frac{1}{2}}. \label{g22}
		\end{eqnarray}
		Since $\varepsilon$ is small and $L>2\beta \geq 2\beta_{j}-\varepsilon$, one has that 
		\begin{equation*}
			e^{-\frac{1}{2}(L-2\beta_{j}- C\varepsilon)q_n} r_{j+\frac{1}{2}}\leq \frac{1}{2}r_{j+\frac{1}{2}}.
		\end{equation*}
		Therefore (\ref{g22}) becomes
		\begin{equation}\label{g23}
			r_{j+\frac{1}{2}}\leq  \exp\{-\frac{1}{2}(L-2\beta_{j}- C\varepsilon)q_n\} (r_{j}+r_{j+1}).
		\end{equation}
		
		{\bf Case 2:} $j_0\in [jq_n+\frac{1}{4}q_n,jq_n+\frac{3}{8}q_n]\cap I_2$
		
		Let $x=jq_n+\frac{3}{8}q_n-j_0$. Therefore, $0\leq x\leq \frac{1}{8}q_n$,
		$$x_1=jq_n-x\in[jq_n-\frac{1}{8}q_n, jq_n]$$ and 
		$$x_2=jq_n+\frac{3}{4} q_n-x\in [jq_n+\frac{5}{8}q_n, jq_n+\frac{3}{4}q_n] .$$
		\begin{center}
			\begin{tikzpicture}[thick, scale=1]
				\draw[-](0,0)--(12,0); 		\draw[-](7.25,0)--(12,0); 	
				\draw[red,-](1,0)--(2.25,0); 
				\draw[blue,-](6,0)--(4.75,0); 
				\draw[red,-](8.5,0)--(9.75,0); 
				\draw[-](6,0)--(6,-0.2);
				\node [below]  at (6,-0.2){$jq_n+\frac{3q_n}{8}$};
				\draw[-](1,0)--(1,0.2);
				\node [above]  at (1,0.2){$-\frac{q_n}{8}$};
				\draw[-](2.25,0)--(2.25,-0.2);
				\node [below]  at (2.25,-0.2){$jq_n$};
				\draw[-](11,0)--(11,-0.2);
				\node [below]  at (11,-0.2){$jq_n+\frac{7q_n}{8}$};
				\draw[-](3.5,0)--(3.5,0.2);
				\node [above]  at (3.5,0.2){$\frac{q_n}{8}$};
				\draw[-](8.5,0)--(8.5,0.2);
				\node [above]  at (8.5,0.2){$\frac{5q_n}{8}$};
				\draw[-](9.75,0)--(9.75,0.2);
				\node [above]  at (9.75,0.2){$\frac{3q_n}{4}$};
				
				\draw[-](4.75,0)--(4.75,0.2);
				\node [above]  at (4.75,0.2){$\frac{q_n}{4}$};
				\draw[-](7.25,0)--(7.25,0.2);
				\node [above]  at (7.25,0.2){$\frac{q_n}{2}$};
				\node [above]  at (5.2,1){$j_0$};
				\draw[->] (5.2,1)--(5,0);
				
				\node [above]  at (1.5,1){$x_1$};
				\draw[->] (1.5,1)--(1.3,0);
				\node [above]  at (8.6,1){$x_2$};
				\draw[dashed] (8.8,0)--(8.8,-0.5);
				\draw[dashed] (9.75,0)--(9.75,-0.5);
				\node [below]  at (9.3,-0.05){$x$};
				\draw [->]  (9.4,-0.3)--(9.75,-0.3);
				\draw [->]  (9.2,-0.3)--(8.8,-0.3);
				\draw[->] (8.6,1)--(8.8,0);

				\node [below]  at (1.8,-0.05){$x$};

				\node [below]  at (5.5,-0.05){$x$};
				
				\node [below] at (6,-1){Fig.2};

			\end{tikzpicture}
		\end{center}

		In Fig.2, $x_1$ and $x_2$ locate at the red part. $j_0$ locates at the blue part. The  numbers above are  the sites after deducting  $jq_n$.

		By Lemma \ref{Le.resonant}, one has that
		\begin{equation}\label{g7}
			|\phi(x_1^\prime)\leq r_{j} e^{-L  x+C\varepsilon q_n} +r_{j-\frac{1}{2}} e^{-L (\frac{q_n}{2} -x)+C\varepsilon q_n} 
		\end{equation}
		and 
		\begin{equation}\label{g8}
			|\phi(x_2^\prime)\leq r_{j+\frac{1}{2}} e^{-L (\frac{q_n}{4}-x) +C\varepsilon q_n} +r_{j+1}e^{-L (\frac{q_n}{4}+x) +C\varepsilon q_n} .
		\end{equation}
		
		By \eqref{g0}  and \eqref{g7}, one has
		\begin{align}
			e^{-\frac{3}{4}Lq_n+\beta_{j} q_n+C\varepsilon q_n }  &
			|P_{[p+1,x_2]} | |\phi (x_1^\prime)|	\nonumber	\\
			\leq   &   	e^{-\frac{3}{4}Lq_n+\beta_{j} q_n +C\varepsilon q_n}  e^{ L(\frac{q_n}{4}-x)}   ( r_{j} e^{-L  x} +r_{j-\frac{1}{2}} e^{-L (\frac{q_n}{2} -x)} )   	\nonumber	\\
			\leq &	e^{-\frac{1}{2}Lq_n-2Lx+\beta_{j} q_n+ C\varepsilon q_n }
			r_{j}+ e^{- L q_n+\beta_{j} q_n+  C\varepsilon q_n}  r_{j-\frac{1}{2}} \nonumber\\
			\leq &    e^{-\frac{1}{2}(L-2\beta_{j}- C\varepsilon)q_n}r_{j} + e^{- L q_n+\beta_{j} q_n+  C\varepsilon q_n}  r_{j-\frac{1}{2}}.\label{g9}
		\end{align}
		By \eqref{g500}, one has 
		\begin{equation}\label{g19}
			r_{j-\frac{1}{2}}\leq e^{C\varepsilon q_n+\frac{q_n}{2} L} r_j. 
		\end{equation}
		By  \eqref{g9} and \eqref{g19}, one has
		\begin{equation}\label{g21}
			e^{-\frac{3}{4}Lq_n+\beta_{j} q_n+C\varepsilon  q_n }  
			|P_{[p+1,x_2]}| |\phi (x_1^\prime)|\leq   e^{-\frac{1}{2}(L-2\beta_{j}- C\varepsilon)q_n} r_{j} .
		\end{equation}
		By Lemma \ref{klem2}
		\begin{equation}\label{g16}
			|P_{[x_1-jq_n,p-1-jq_n]}| \leq e^{C\varepsilon q_n} e^{L(\frac{1}{2}q_n-x)}.
		\end{equation}
		By Lemma \ref{klem1}  and \eqref{g16}, one has
		\begin{align}
			|P_{[x_1,p-1]}| &\leq  	|P_{[x_1,p-1]}-P_{[x_1-jq_n,p-1-jq_n]}|+|P_{[x_1-jq_n,p-1-jq_n]}|\nonumber\\
			&\leq e^{C\varepsilon q_n} e^{L(\frac{1}{2}q_n-x)}+e^{C\varepsilon q_n} e^{L |p-x_1|} e^{-\beta_{j} q_n}\nonumber\\
			&\leq 	e^{C\varepsilon q_n} e^{L(\frac{1}{2}q_n-x)}+e^{C\varepsilon q_n} e^{L(\frac{1}{2}q_n+x)  -\beta_{j} q_n}.
			\label{g17}
		\end{align}
		By \eqref{g8} and \eqref{g17},  one has
		\begin{align}
			e^{-\frac{3}{4}Lq_n+\beta_{j} q_n+C\varepsilon q_n } &|P_{[x_1,p-1]}| |\phi (x_2^\prime)| \nonumber\\
			\leq& e^{-\frac{3}{4}Lq_n+\beta_{j} q_n+C\varepsilon q_n }(	  r_{j+\frac{1}{2}} e^{-L (\frac{q_n}{4}-x) } +r_{j+1}e^{-L (\frac{q_n}{4}+x)} )  \nonumber \\
			& \times (e^{L(\frac{1}{2}q_n-x)}+ e^{L(\frac{1}{2}q_n+x)  -\beta_{j} q_n}) \nonumber\\
			\leq&  e^{C\varepsilon q_n}(e^{-(L-2\beta_{j})\frac{q_n}{2}  }+e^{  -L(\frac{q_n}{2}-2x)}) r_{j+\frac{1}{2}}\nonumber \\
			&+ e^{C\varepsilon q_n}(e^{-(L-2\beta_{j})\frac{q_n}{2}  -2L x}+e^{  -L\frac{q_n}{2}}) r_{j+1}\nonumber\\
			\leq & e^{C\varepsilon q_n}(e^{-(L-2\beta_{j})\frac{q_n}{2}  }+e^{  -L\frac{q_n}{4}}) r_{j+\frac{1}{2}}+  e^{-(L-2\beta_{j}-C\varepsilon)\frac{q_n}{2}  }  r_{j+1}\label{g18},
		\end{align}
		where the last inequality holds  by the fact that $0\leq x\leq \frac{q_n}{8}$.
		
		By \eqref{g2}, \eqref{g21} and \eqref{g18}, one has
		
		\begin{eqnarray}
			r_{j+\frac{1}{2}} &\leq&  e^{-\frac{1}{2}(L-2\beta_{j}- C\varepsilon)q_n} r_{j}+e^{-\frac{1}{2}(L-2\beta_{j}- C\varepsilon)q_n}r_{j+1} \nonumber\\
			& &+ e^{C\varepsilon q_n}(e^{-(L-2\beta_{j})\frac{q_n}{2}  }+e^{  -L\frac{q_n}{4}}) r_{j+\frac{1}{2}}. \label{g24}
		\end{eqnarray}
		
		Since $\varepsilon$ is small and $L>2\beta \geq 2\beta_{j}-\varepsilon$, one has that 
		\begin{equation*}
			e^{C\varepsilon q_n}(e^{-(L-2\beta_{j})\frac{q_n}{2}  }+e^{  -L\frac{q_n}{4}})  r_{j+\frac{1}{2}}\leq \frac{1}{2}	r_{j+\frac{1}{2}}.
		\end{equation*}
		
		Therefore,  (\ref{g24}) becomes
		\begin{equation}\label{g25}
			r_{j+\frac{1}{2}}\leq  \exp\{-\frac{1}{2}(L-2\beta_{j}- C\varepsilon)q_n\}(r_j+r_{j+1}).
		\end{equation}
		
		{\bf Case 3: } $j_0\in [jq_n + \frac{5}{8}q_n,jq_n+\frac{3}{4}q_n]\cap I_2$

		Let $x=j_0-jq_n -\frac{5}{8}q_n$. Therefore, $0\leq x\leq \frac{1}{8}q_n$,  $x_1=jq_n +\frac{q_n}{4}+x$ and $x_2=jq_n+ q_n+x$.
		
		By Lemmas \ref{klem2} and  \ref{klem1},  one has
		\begin{align}
			|	P_{[p+1,x_2]} | 
			&\leq 	|	P_{{[p+1-jq_n-q_n,x_2-jq_n-q_n]}} | +e^{C\varepsilon q_n} e^{L(\frac{1}{2}q_n+x)  -\beta_{j} q_n} \nonumber\\
			&\leq e^{C\varepsilon q_n} e^{L(\frac{1}{2}q_n-x)}+e^{C\varepsilon q_n} e^{L(\frac{1}{2}q_n+x)  -\beta_{j} q_n}	\label{g171}
		\end{align}
		Replacing \eqref{g17} with \eqref{g171} and following the proof of Case 2, we also have that 
		\begin{equation}\label{g26}
			r_{j+\frac{1}{2}}\leq  \exp\{-\frac{1}{2}(L-2\beta_{j}- C\varepsilon)q_n\}(r_j+r_{j+1}).
		\end{equation}

		If $j_0\in I_1$,  one has that 
		$$x_1\in[-\frac{5}{8}q_n,-\frac{3}{8}q_n],x_2\in [\frac{1}{8}q_n, \frac{3}{8}q_n].$$
		By \eqref{Numerator1}, one has
		\begin{equation}\label{gg100}
			|P_{[x_1,-1]}| \leq e^{ L|x_1|+C\varepsilon q_n},  	|P_{[1,x_2]} | \leq e^{ L|x_2|+C\varepsilon q_n}.
		\end{equation}
		By  \eqref{g100},  \eqref{g1091}, \eqref{sq},  \eqref{g68} and \eqref{gg100}, one has 
		\begin{align}
			|\phi(0)| \leq& e^{-kL +\varepsilon q_n } (|P_{[x_1,-1]}| |\phi (x_2^\prime)|+
			|P_{[1,x_2]} | |\phi (x_1^\prime)|)\nonumber\\
			\leq &e^{-\frac{3}{4}Lq_n+C\varepsilon  q_n } (e^{L|x_1|} +e^{L|x_2|})\nonumber\\
			\leq &\frac{1}{2} \nonumber.
		\end{align}	
		This is impossible since $\phi(0)=1$. Therefore, we must have  $j_0\in I_2$ and Theorem \ref{thm1} follows from \eqref{g23}, \eqref{g25} and \eqref{g26}.

	\end{proof}
	
	\section{Resonance II: sites $jq_n$, $j\in\Z\backslash \{0\}$}\label{S5}
	In this section, we   deal with  resonances arising from sites $jq_n$, $j\in\Z\backslash\{0\}$.
	\begin{theorem}\label{thm2}
		%Suppose $b_{n+1}\geq \frac{q_n}{4}$.
		Assume $j\in\Z$ satisfies $ 0<  |j |\leq 2\frac{b_{n+1}}{q_n}+10$. Then  we have
		%Suppose  $\frac{b_n}{4}\leq |k|\leq b_{n+1}$ is resonant and $|k|\geq \frac{q_n}{2}$, then
		\begin{equation}\label{g1000}
			r_{j}  \leq   \exp\{- \frac{1}{2}(L-2\beta_{j} -C\varepsilon)q_n\}(r_{j+\frac{1}{2}}+r_{j-\frac{1}{2}})+ \exp\{- (L-2\beta_{j} -C\varepsilon)q_n\}(r_{j+1}+r_{j-1}).
		\end{equation}
		
	\end{theorem}	
	
	\begin{proof}
		Take  $\phi(j q_n+rq_n)$ with $|r|\leq 10\varepsilon $ into consideration. Denote by $p=jq_n+rq_n$.
		Without loss of generality assume $j>0$.
		Let $n_0$ be the least positive integer such that
		\begin{equation*}
			q_{n-n_0}\leq  \frac{ \varepsilon}{2}  (\frac{1}{4}-2\varepsilon) q_n.
		\end{equation*}
		Let $s$ be the
		largest positive integer such that $sq_{n-n_0}\leq (\frac{1}{4}-2\varepsilon) q_n $.
		By the fact that $(s+1)q_{n-n_0}\geq (\frac{1}{4}-2\varepsilon)q_n $, one has 
		\begin{equation*}
			s\geq \frac{1}{\varepsilon} 
		\end{equation*}
		and
		\begin{equation}\label{2sq}
			(\frac{1}{4}-3\varepsilon)q_n \leq sq_{n-n_0}\leq (\frac{1}{4}-2\varepsilon)q_n.
		\end{equation}

		%For simplicity, we assume $n_0=1$.
		Set $I_1, I_2\subset \mathbb{Z}$ as follows
		\begin{eqnarray*}
			I_1 &=& [- 2s q_{n-n_0},-1], \\
			I_2 &=& [ j q_n  -2sq_{n-n_0},j q_n+ 2sq_{n-n_0}-1 ],
		\end{eqnarray*}
		and let $\theta_m=\theta+m\alpha$ for $m\in I_1\cup I_2$. The set $\{\theta_m\}_{m\in I_1\cup I_2}$
		consists of $6sq_{n-n_0}$ elements. Let $k=6s q_{n-n_0}-1$.
		
		%Considering set $I_1\cup I_2$,
		By  modifying  the proof of ~\cite[Lemma 9.9]{aj09} and   ~\cite[Lemma 4.1]{lyjfa} (or Appendices in ~\cite{jl18} and ~\cite{liuetds20}), we can prove the claim (Claim 2):
		for any
		$  m\in I_1\cup I_2$, one has ${\rm Lag}a_m\leq  2q_n (\beta_{j}+ \varepsilon) $.
		We also give the proof in the Appendix.
		
		% and $C$  is a large absolute constant below.
		Applying   Lemma  \ref{Le.Uniform}, there exists some $j_0$ with  $j_0\in I_1\cup I_2$
		such that
		\begin{equation}\label{g121}
			|P_k(\theta_{j_0}-\frac{k-1}{2}\alpha)| \geq e^{kL-2(\beta_{j} +\varepsilon)q_n}.
		\end{equation}

		Let 
		\begin{equation}\label{2g4}
			I=[j_0-3sq_{n-n_0} +1,j_0+3sq_{n-n_0}-1]=[x_1,x_2]. 
		\end{equation}

		%	Denote by   $x_1^\prime=x_1-1$ and $x_2^\prime=x_2+1$.

		By  \eqref{g100}, \eqref{g121}  and \eqref{2sq},   one has
		\begin{itemize}
			\item I: $j_0\in I_1$.  Then for any $y_1$ with $|y_1|\leq 10 \varepsilon q_n$,
			\begin{align}
				|\phi(y_1)| \leq &e^{-kL+2\beta_{j} q_n+\varepsilon q_n } (|P_{[x_1,y_1-1]}| |\phi (x_2^\prime)|+
				|P_{[y_1+1,x_2]} | |\phi (x_1^\prime)|)\nonumber\\
				\leq &e^{-\frac{3}{2}Lq_n+2\beta_{j} q_n+ C\varepsilon q_n } (|P_{[x_1,y_1-1]}| |\phi (x_2^\prime)|+
				|P_{[y_1+1,x_2]} | |\phi (x_1^\prime)|).\label{12g2}
			\end{align}
			
			\item  II: $j_0\in I_2$.  Then for any $y_2$ with $|y_2-jq_n|\leq 10 \varepsilon q_n$,
			\begin{equation}\label{2g2}
				|\phi(y_2)| \leq e^{-\frac{3}{2}Lq_n+2\beta_{j} q_n+C\varepsilon q_n } (|P_{[x_1,y_2-1]}| |\phi (x_2^\prime)|+
				|P_{[y_2+1,x_2]} | |\phi (x_1^\prime)|).
			\end{equation}
		\end{itemize}
		Clearly,
		\begin{equation*}
			I_1\subset [-\frac{q_n}{2},0]
		\end{equation*}
		and
		\begin{equation*}
			I_2\subset [jq_n-\frac{q_n}{2} ,jq_n+\frac{q_n}{2}].
		\end{equation*}
		{\bf Case 1: } $j_0\in [jq_n-\frac{1}{4} q_n, jq_n]\cap I_2$
		
		%In this case, by \eqref{2sq} and \eqref{2g4}, one has  that
		
		Let $x=j_0-(jq_n-\frac{1}{4}q_n)$. Therefore, $0\leq x\leq \frac{1}{4}q_n$ and
		$$ x_1=jq_n-q_n+x \in[jq_n-q_n,jq_n-\frac{3}{4}q_n]$$ and 
		$$x_2=jq_n+\frac{1}{2}q_n+x \in [jq_n+\frac{1}{2}q_n,jq_n+\frac{3}{4}q_n].$$
		\begin{center}
			\begin{tikzpicture}[thick, scale=1]
				\draw[-](0,0)--(12,0); 		\draw[-](7.25,0)--(12,0); 	
				\draw[red,-](1,0)--(2.25,0); 
				\draw[blue,-](6,0)--(4.75,0); 
				\draw[red,-](8.5,0)--(9.75,0); 
				\draw[-](6,0)--(6,-0.2);
				\node [below]  at (6,-0.2){$jq_n$};
				\draw[-](1,0)--(1,-0.2);
				\node [below]  at (1,-0.2){$jq_n-q_n$};
				\draw[-](2.25,0)--(2.25,0.2);
				\node [above]  at (2.25,0.2){$-\frac{3q_n}{4}$};
				\draw[-](11,0)--(11,-0.2);
				\node [below]  at (11,-0.2){$jq_n+q_n$};
				\draw[-](3.5,0)--(3.5,0.2);
				\node [above]  at (3.5,0.2){$-\frac{q_n}{2}$};
				\draw[-](8.5,0)--(8.5,0.2);
				\node [above]  at (8.5,0.2){$\frac{q_n}{2}$};
				\draw[-](9.75,0)--(9.75,0.2);
				\node [above]  at (9.75,0.2){$\frac{3q_n}{4}$};
				
				\draw[-](4.75,0)--(4.75,0.2);
				\node [above]  at (4.75,0.2){$-\frac{q_n}{4}$};
				\draw[-](7.25,0)--(7.25,0.2);
				\node [above]  at (7.25,0.2){$\frac{q_n}{4}$};
				\node [above]  at (5.2,1){$j_0$};
				\draw[->] (5.2,1)--(5.5,0);
				
				\node [above]  at (1.5,1){$x_1$};
				\draw[->] (1.5,1)--(1.8,0);
				\node [above]  at (9.5,1){$x_2$};
				\draw[->] (9.5,1)--(9.3,0);
				
				\node [below]  at (9,0){$x$};

				\node [above]  at (1.4,0){$x$};

				\node [above]  at (5.1,0){$x$};
				
				\node [below] at (6,-1){Fig.3};

			\end{tikzpicture}
		\end{center}

		In Fig.3, $x_1$ and $x_2$ locate at the red part. $j_0$ locates at the blue part. The  numbers above are  the sites after deducting  $jq_n$.

		By Lemma \ref{Le.resonant}, one has that
		\begin{equation}\label{2g7}
			|\phi(x_1^\prime)\leq r_{j-1} e^{-L  x+C\varepsilon q_n} +r_{j-\frac{1}{2}} e^{-L (\frac{q_n}{2} -x)+C\varepsilon q_n} 
		\end{equation}
		and 
		\begin{equation}\label{2g8}
			|\phi(x_2^\prime)\leq r_{j+\frac{1}{2}} e^{-L x +C\varepsilon q_n} +r_{j+1}e^{-L (\frac{q_n}{2}-x) +C\varepsilon q_n} .
		\end{equation}
		By \eqref{Numerator1}, one has 
		\begin{equation}\label{2g1}
			|P_{[x_1,p-1]}|\leq e^{ L(q_n-x)+C\varepsilon q_n},  
		\end{equation}
		and
		\begin{equation}\label{2g0}
			|P_{[p+1,x_2]}|\leq e^{ L(\frac{q_n}{2}+x)+C\varepsilon q_n}.
		\end{equation}
		
		By \eqref{2g2}-\eqref{2g0},
		one has  that
		\begin{align}
			|\phi(p)|\leq& e^{-\frac{3}{2}Lq_n+2\beta_{j} q_n+C\varepsilon q_n }	|P_{[x_1,p-1]}|(r_{j+\frac{1}{2}} e^{-L x } +r_{j+1}e^{-L (\frac{q_n}{2}-x)} )\nonumber\\
			&+e^{-\frac{3}{2}Lq_n+2\beta_{j} q_n+C\varepsilon q_n }	|P_{[p+1,x_2]}|(r_{j-1} e^{-L  x} +r_{j-\frac{1}{2}} e^{-L (\frac{q_n}{2} -x)}  ) \nonumber\\
			\leq& e^{-\frac{3}{2}Lq_n+2\beta_{j} q_n+C\varepsilon q_n }(	|P_{[x_1,p-1]}|r_{j+\frac{1}{2}} e^{-L x } + e^{L(q_n-x)}r_{j+1}e^{-L (\frac{q_n}{2}-x)} )\nonumber\\
			&+e^{-\frac{3}{2}Lq_n+2\beta_{j} q_n+C\varepsilon q_n }e^{ L(\frac{q_n}{2}+x)}(r_{j-1} e^{-L  x} +r_{j-\frac{1}{2}} e^{-L (\frac{q_n}{2} -x)}  ) \nonumber\\
			\leq &e^{-L q_n+2\beta_{j} q_n+C\varepsilon q_n} (r_{j+1}+r_{j-1})+e^{-\frac{3}{2}L q_n+2\beta_{j} q_n +2L x+C\varepsilon q_n} r_{j-\frac{1}{2}}\nonumber\\
			&+e^{-\frac{3}{2}L q_n-L x+2\beta_{j} q_n+C\varepsilon q_n}|P_{[x_1,p-1]}|r_{j+\frac{1}{2}}\nonumber\\
			\leq& e^{-L q_n+2\beta_{j} q_n+C\varepsilon q_n} (r_{j+1}+r_{j-1}+r_{j-\frac{1}{2}})\label{g501}\\
			&+e^{-\frac{3}{2}L q_n-L x+2\beta_{j} q_n+C\varepsilon q_n}|P_{[x_1,p-1]}|r_{j+\frac{1}{2}},\label{g41}
		\end{align}
		where the last inequality holds by the fact that $0\leq x\leq \frac{1}{4}q_n$.
		
		The  term in \eqref{g501} decays, so we are going to bound \eqref{g41}.
		
		\begin{center}
			\begin{tikzpicture}[thick, scale=1.5]
				\draw[-](0,0)--(7,0); 		 
				\draw[red,-](1,0)--(2.25,0);

				\draw[-](6,0)--(6,-0.2);
				\node [below]  at (6,-0.2){$jq_n+q_n$};
				\node [above]  at (6,0){$p+q_n$};
				\draw[-](1,0)--(1,-0.2);
				\node [below]  at (1,-0.2){$jq_n$};
				\draw[-](2.25,0)--(2.25,-0.2);
				%	\node [below]  at (2.25,-0.2){$jq_n$};

				\draw[-](3.5,0)--(3.5,-0.2);
				\node [above]  at (3.5,0){$y$};
				\node [above]  at (3.5,-0.6){$jq_n+\frac{q_n}{2}$};

				\draw[-](4.75,0)--(4.75,-0.2);
				\node [above]  at (4.75,-0.6){$jq_n+\frac{3q_n}{4}$};

				\node [above]  at (1.5,1){$x_1+q_n$};
				\draw[->] (1.5,1)--(1.8,0);

				\node [above]  at (1.4,0){$x$};

				\node [below] at (3,-1){Fig.4};

			\end{tikzpicture}
		\end{center}

		Clearly, $x_1+q_n\in[jq_n, jq_n+\frac{1}{4}q_n]$. 
		By Lemma \ref{Le.resonant}, one has that
		\begin{equation}\label{2g30}
			|\phi(x_1'+q_n)\leq r_{j} e^{-L  x+C\varepsilon q_n} +r_{j+\frac{1}{2}} e^{-L (\frac{q_n}{2} -x)+C\varepsilon q_n} .
		\end{equation}
		By \eqref{g100},  \eqref{Numerator1} and \eqref{2g30}, one has 
		for  any $y$ with $|y-jq_n-\frac{1}{2} q_n|\leq 10 \varepsilon q_n$,
		\begin{align*}
			|\phi(y)|\leq&  |P_{[x_1+q_n,p-1+q_n]}|^{-1}  e^{L \frac{q_n}{2}+C\varepsilon q_n} (r_{j} e^{-L  x} +r_{j+\frac{1}{2}} e^{-L (\frac{q_n}{2} -x)})\\
			&+|P_{[x_1+q_n,p-1+q_n]}|^{-1} e^{L(\frac{q_n}{2}-x+C\varepsilon q_n)} r_{j+1}.
		\end{align*}
		This implies
		\begin{align*}
			r_{j+\frac{1}{2}}\leq&  |P_{[x_1+q_n,p-1+q_n]}|^{-1}  e^{L \frac{q_n}{2}+C\varepsilon q_n} (r_{j} e^{-L  x} +r_{j+\frac{1}{2}} e^{-L (\frac{q_n}{2} -x)})\\
			&+|P_{[x_1+q_n,p-1+q_n]}|^{-1} e^{L(\frac{q_n}{2}-x+C\varepsilon q_n)} r_{j+1}.
		\end{align*}
		Therefore, we  have
		\begin{align}
			r_{j+\frac{1}{2}}&	|P_{[x_1+q_n,p-1+q_n]}|\nonumber\\
			\leq &e^{ C\varepsilon q_n} (r_{j} e^{L(\frac{q_n}{2}- x)} +r_{j+\frac{1}{2}} e^{Lx}+  e^{L(\frac{q_n}{2}-x)} r_{j+1}). \label{g37}
		\end{align}
		
		By Lemma \ref{klem1}, one has 
		\begin{equation}\label{g38}
			| P_{[x_1+q_n,p-1+q_n]}- P_{[{x}_1,p-1]}|\leq e^{(L+C\varepsilon) q_n-Lx-\beta q_n}.
		\end{equation}
		By \eqref{g37} and \eqref{g38}, one has
		\begin{align}
			|P_{[{x}_1,p-1]}|	&r_{j+\frac{1}{2}} \nonumber\\
			\leq & e^{ L(\frac{q_n}{2}-x)+ C\varepsilon q_n} (r_{j}  +   r_{j+1})  + e^{C\varepsilon q_n}( e^{Lx}+e^{(L-\beta) q_n-Lx} )r_{j+\frac{1}{2}}\label{g39}
		\end{align}
		Therefore,
		\begin{align}
			|P_{[x_1,p-1]}|r_{j+\frac{1}{2}}  &	e^{-\frac{3}{2}L q_n-L x+2\beta_{j} q_n+C\varepsilon q_n}\nonumber\\
			\leq & e^{C\varepsilon q_n}(e^{- (L-2\beta_{j}  ) \frac{1}{2}q_n-2L x} +e^{-\frac{3}{2}L q_n+2\beta_{j} q_n+C\varepsilon q_n}) r_{j+\frac{1}{2}} \nonumber\\
			&+e^{-L q_n+2\beta_{j} q_n-2L x+C\varepsilon q_n}(r_j+r_{j+1})\nonumber\\
			\leq &e^{C\varepsilon q_n} (r_j e^{-(L-2\beta _{j} ) q_n}+r_{j+\frac{1}{2}} e^{-(L-2\beta _{j} ) \frac{q_n}{2}}+ r_{j+1} e^{-(L-2\beta _{j} ) q_n})\label{g40}
		\end{align}
		
		By \eqref{g40}, \eqref{g501} and \eqref{g41}, one has 
		\begin{equation}\label{g42}
			|\phi(p)| \leq  e^{-(L-2\beta _{j} ) q_n+C\varepsilon q_n} (r_{j-1}+r_{j-\frac{1}{2}}+r_j+r_{j+1})+r_{j+\frac{1}{2}} e^{-(L-2\beta _{j} -C\varepsilon) \frac{q_n}{2}}.
		\end{equation}
		Therefore,
		\begin{equation}\label{g43}
			r_{j}\leq  e^{-(L-2\beta _{j} ) q_n+C\varepsilon q_n} (r_{j-1}+r_{j-\frac{1}{2}}+r_j+r_{j+1})+r_{j+\frac{1}{2}} e^{-(L-2\beta _{j} -C\varepsilon) \frac{q_n}{2}}.
		\end{equation}
		By the fact that $L >2\beta \geq 2\beta_{j}-\varepsilon$, we have 
		$$ e^{-(L-2\beta _{j} ) q_n+C\varepsilon q_n} r_j\leq\frac{1}{2}r_j $$
		and hence
		\begin{equation}\label{g44}
			r_{j}\leq  e^{-(L-2\beta _{j} -C\varepsilon) q_n} (r_{j-1}+r_{j-\frac{1}{2}}+r_{j+1})+r_{j+\frac{1}{2}} e^{-(L-2\beta _{j} -C\varepsilon) \frac{q_n}{2}}.
		\end{equation}
		
		{\bf Case 2}:	   $j_0\in [jq_n, jq_n+\frac{1}{4} q_n]\cap I_2$
		
		In this case, by the similar  proof of \eqref{g44}, we  have
		
		\begin{equation}\label{g45}
			r_{j}\leq  e^{-(L-2\beta _{j} -C\varepsilon) q_n} (r_{j+1}+r_{j+\frac{1}{2}}+r_{j-1})+r_{j-\frac{1}{2}} e^{-(L-2\beta _{j} -C\varepsilon) \frac{q_n}{2}}.
		\end{equation}
		
		{\bf Case 3}:		$j_0\in[jq_n-\frac{1}{2}q_n, jq_n-\frac{1}{4}q_n]\cap I_2$
		
		In this case, let $x=j_0-(jq_n-\frac{1}{2}q_n)$.  It is easy to see  that $0\leq x\leq \frac{q_n}{4}$,
		$$x_1=jq_n-\frac{5q_n}{4}+x \in[jq_n-\frac{5}{4}q_n, jq_n-q_n]$$ and 
		$$x_2=jq_n+\frac{1}{4}q_n+x\in [jq_n+\frac{q_n}{4}, jq_n+\frac{q_n}{2}].$$

		\begin{center}
			\begin{tikzpicture}[thick, scale=1]
				\draw[-](0,0)--(12,0); 		\draw[-](7.25,0)--(12,0); 	
				\draw[red,-](1,0)--(2.25,0); 
				\draw[blue,-](6,0)--(4.75,0); 
				\draw[red,-](8.5,0)--(9.75,0); 
				\draw[-](6,0)--(6,0.2);
				\node [above]  at (6,0.2){$-\frac{q_n}{4}$};
				\draw[-](1,0)--(1,0.2);
				\node [above]  at (1,0.2){$-\frac{5q_n}{4}$};
				\draw[-](2.25,0)--(2.25,-0.2);
				\node [below]  at (2.25,-0.2){$jq_n-q_n$};
				
				\draw[-](3.5,0)--(3.5,0.2);
				\node [above]  at (3.5,0.2){$-\frac{3q_n}{4}$};
				\draw[-](8.5,0)--(8.5,0.2);
				\node [above]  at (8.5,0.2){$\frac{q_n}{4}$};
				\draw[-](9.75,0)--(9.75,0.2);
				\node [above]  at (9.75,0.2){$\frac{q_n}{2}$};
				
				\draw[-](4.75,0)--(4.75,0.2);
				\node [above]  at (4.75,0.2){$-\frac{q_n}{2}$};
				\draw[-](7.25,0)--(7.25,-0.2);
				\node [below]  at (7.25,-0.2){$jq_n$};
				\node [above]  at (5.2,1){$j_0$};
				\draw[->] (5.2,1)--(5.6,0);
				
				\node [above]  at (1.5,1){$x_1$};
				\draw[->] (1.5,1)--(1.8,0);
				\node [above]  at (9.6,1){$x_2$};
				\draw[->] (9.6,1)--(9.4,0);
				
				\node [above]  at (9.1,0){$x$};

				\node [above]  at (1.4,0){$x$};

				\node [above]  at (5.3,0){$x$};
				
				\node [below] at (6,-1){Fig.5};

			\end{tikzpicture}
		\end{center}

		By Lemma \ref{Le.resonant}, one has that
		\begin{equation}\label{2g44}
			|\phi(x_1^\prime)\leq r_{j-1} e^{-L  (\frac{q_n}{4}-x)+C\varepsilon q_n} +r_{j-\frac{3}{2}} e^{-L (\frac{q_n}{4} +x)+C\varepsilon q_n} 
		\end{equation}
		and 
		\begin{equation}\label{2g45}
			|\phi(x_2^\prime)\leq r_{j} e^{-L (\frac{q_n}{4} +x) +C\varepsilon q_n} +r_{j+\frac{1}{2}}e^{-L (\frac{q_n}{4}-x) +C\varepsilon q_n} .
		\end{equation}
		
		By   \eqref{2g2},  \eqref{2g44} and \eqref{2g45}, one has
		\begin{align}
			r_j&\leq e^{-\frac{3}{2} L q_n+2\beta_{j} q_n+C\varepsilon q_n}|P_{[p+1,x_2]}| (e^{-L (\frac{q_n}{4}-x) } r_{j-1}+ e^{-L (\frac{q_n}{4}+x) } r_{j-\frac{3}{2}})\nonumber\\
			&+e^{-\frac{3}{2} L q_n+2\beta_{j} q_n+C\varepsilon q_n}|P_{[x_1,p-1]}|  (e^{-L (\frac{q_n}{4}+x) } r_{j}+ e^{-L (\frac{q_n}{4}-x) } r_{j+\frac{1}{2}}).\label{g63}
		\end{align}
		
		We are going to bound 
		\begin{equation}\label{gg102}
			e^{-\frac{3}{2} L q_n+2\beta_{j} q_n+C\varepsilon q_n}|P_{[p+1,x_2]}| (e^{-L (\frac{q_n}{4}-x) } r_{j-1}+ e^{-L (\frac{q_n}{4}+x) } r_{j-\frac{3}{2}})
		\end{equation}
		and 
		\begin{equation}\label{gg101}
			e^{-\frac{3}{2} L q_n+2\beta_{j} q_n+C\varepsilon q_n}|P_{[x_1,p-1]}|  (e^{-L (\frac{q_n}{4}+x) } r_{j}+ e^{-L (\frac{q_n}{4}-x) } r_{j+\frac{1}{2}})
		\end{equation}
		separately.
		
		By  \eqref{Numerator1}, one has
		\begin{equation}\label{g47}
			|P_{[p+1,x_2]}|\leq e^{L(\frac{1}{4}q_n+x) +C\varepsilon q_n}.
		\end{equation}
		By   \eqref{g47}, one has
		\begin{align}
			e^{-\frac{3}{2} L q_n+2\beta_{j} q_n} &|P_{[p+1,x_2]}|  (e^{-L (\frac{q_n}{4}-x) } r_{j-1}+ e^{-L (\frac{q_n}{4}+x) } r_{j-\frac{3}{2}})\nonumber\\
			%	\leq %	& 	e^{-\frac{3}{2} L q_n+2\beta_{j} q_n+C\varepsilon q_n} |P_{[p+1,x_2]}|	 (r_{j-1} e^{-L  (\frac{q_n}{4}-x)} +r_{j-\frac{3}{2}} e^{-L (\frac{q_n}{4} +x)}) \nonumber\\
			%	\leq &   e^{-\frac{3}{2} L q_n+2\beta_{j} q_n+C\varepsilon q_n}  e^{L(\frac{1}{4}q_n+x)} (e^{-L (\frac{q_n}{4}-x) } r_{j-1}+ e^{-L (\frac{q_n}{4}+x) } r_{j-\frac{3}{2}})\nonumber\\
			\leq	&  e^{-\frac{3}{2} L q_n+2\beta_{j} q_n+2Lx+C\varepsilon q_n}   r_{j-1}+ e^{-\frac{3}{2} L q_n+2\beta_{j} q_n+C\varepsilon q_n}  r_{j-\frac{3}{2}}.\label{g49}
		\end{align}
		By \eqref{g500}, one has 
		\begin{equation}\label{g50}
			r_{j-\frac{3}{2}}\leq e^{C\varepsilon q_n+\frac{q_n}{2} L} r_{j-1}. 
		\end{equation}
		By \eqref{g49}, \eqref{g50} and the fact that $0\leq x\leq \frac{1}{4}q_n$, 
		one has
		\begin{equation}\label{g51}
			e^{-\frac{3}{2} L q_n+2\beta_{j} q_n}|P_{[p+1,x_2]}|  (e^{-L (\frac{q_n}{4}-x) } r_{j-1}+ e^{-L (\frac{q_n}{4}+x) } r_{j-\frac{3}{2}})\leq e^{-(L-2\beta_{j} ) q_n+C\varepsilon q_n} r_{j-1}.
		\end{equation}
		We finish the estimate of \eqref{gg102}.
		Now  we are in the position to bound \eqref{gg101}.

		By \eqref{g102}, one has
		\begin{equation}\label{g52}
			|P_{[x_1-(j-1)q_n,p-1-(j-1)q_n]}|\leq e^{L (\frac{3q_n}{4}+x)+C\varepsilon q_n}.
		\end{equation}
		By Lemma \ref{klem1} and \eqref{g52}, we have
		\begin{equation}\label{g53}
			|P_{[x_1,p-1]}|\leq e^{L (\frac{3q_n}{4}+x)+C\varepsilon q_n}+e^{L(\frac{5q_n}{4}-x)-\beta_{j}q_n+C\varepsilon q_n}.
		\end{equation}
		By \eqref{g53}, one has
		\begin{align}
			e^{-\frac{3}{2} L q_n+2\beta_{j} q_n} &e^{-L ( \frac{q_n}{4}+x) }|P_{[x_1,p-1]}|  r_{j}\nonumber\\
			\leq &	e^{-\frac{3}{2} L q_n+2\beta_{j} q_n+C\varepsilon q_n}(e^{L (\frac{3q_n}{4}+x)}+e^{L(\frac{5q_n}{4}-x)-\beta_{j}q_n }) e^{-L (\frac{q_n}{4}+x) } r_{j}\nonumber\\
			\leq &e^{-(L-2\beta_{j})q_n+C\varepsilon q_n} r_j+e^{-(L-2\beta_{j})\frac{q_n}{2}-2L x+C\varepsilon q_n} r_j\nonumber\\
			\leq &  e^{-(L-2\beta_{j})\frac{q_n}{2}+C\varepsilon q_n} r_j.\label{g55}
		\end{align}
		
		This implies that the first term in \eqref{gg101} decays. We are going to bound the remaining term in \eqref{gg101}:
		$$	e^{-\frac{3}{2} L q_n+2\beta_{j} q_n}   e^{-L (\frac{q_n}{4}-x) }|P_{[x_1,p-1]}|   r_{j+\frac{1}{2}}.$$

		By Lemma \ref{klem1}, one has 
		\begin{equation}\label{2g38}
			 |P_{[x_1+q_n,p-1+q_n]}- P_{[{x}_1,p-1]}| \leq e^{L (\frac{5}{4}q_n-x)-\beta q_n+C\varepsilon q_n}.
		\end{equation}
		By \eqref{2g38},  we have 
		\begin{align}
			e^{-\frac{3}{2} L q_n+2\beta_{j} q_n}   &e^{-L (\frac{q_n}{4}-x) }|P_{[x_1,p-1]}|   r_{j+\frac{1}{2}}\nonumber\\
			\leq& e^{-\frac{3}{2} L q_n+2\beta_{j} q_n}|P_{[x_1+q_n,p-1+q_n]}|    e^{-L (\frac{q_n}{4}-x) } r_{j+\frac{1}{2}}\nonumber \\
			&+ e^{L (\frac{5}{4}q_n-x)-\beta q_n+C\varepsilon q_n}e^{-\frac{3}{2} L q_n+2\beta_{j} q_n}    e^{-L (\frac{q_n}{4}-x) } r_{j+\frac{1}{2}} \nonumber\\
			\leq & e^{-\frac{3}{2} L q_n+2\beta_{j} q_n}  e^{-L (\frac{q_n}{4}-x) } |P_{[x_1+q_n,p-1+q_n]}|    r_{j+\frac{1}{2}}  +e^{- (L-2\beta_{j}-C\varepsilon)\frac{q_n}{2}}r_{j+\frac{1}{2}}.\label{g61}
		\end{align}
		The second term in \eqref{g61} decays, so 
		we are going to bound  the first term in \eqref{g61}: $e^{-\frac{3}{2} L q_n+2\beta_{j} q_n}  e^{-L (\frac{q_n}{4}-x) } |P_{[x_1+q_n,p-1+q_n]}|    r_{j+\frac{1}{2}}  $.
		
		Clearly, $x_1+q_n\in[jq_n-\frac{1}{4}q_n, jq_n]$.  	By Lemma \ref{Le.resonant}, one has that
		\begin{equation}\label{22g30}
			|\phi(x_1^\prime+q_n)|\leq r_{j} e^{-L  (\frac{q_n}{4}-x)+C\varepsilon q_n} +r_{j-\frac{1}{2}} e^{-L (\frac{q_n}{4} +x)+C\varepsilon q_n} .
		\end{equation}
		\begin{center}
			\begin{tikzpicture}[thick, scale=1.5]
				\draw[-](0,0)--(8,0); 		 
				\draw[red,-](1,0)--(2.25,0);

				\draw[-](7.25,0)--(7.25,-0.2);
				\node [below]  at (7.25,-0.2){$jq_n+q_n$};
				\node [above]  at (7.25,0){$p+q_n$};
				\draw[-](1,0)--(1,-0.2);
				\node [below]  at (1,-0.2){$jq_n-\frac{q_n}{4}$};
				\draw[-](2.25,0)--(2.25,-0.2);
				\node [below]  at (2.25,-0.2){$jq_n$};
				%	\node [below]  at (2.25,-0.2){$jq_n$};

				\draw[-](3.5,0)--(3.5,-0.2);
				
				\node [above]  at (3.5,-0.6){$jq_n+\frac{q_n}{4}$};

				\node [above]  at (4.75,0){$y$};
				\draw[-](4.75,0)--(4.75,-0.2);
				\node [above]  at (4.75,-0.6){$jq_n+\frac{q_n}{2}$};

				\node [above]  at (1.5,1){$x_1+q_n$};
				\draw[->] (1.5,1)--(1.8,0);

				\node [above]  at (1.4,0){$x$};

				\node [below] at (4,-1){Fig.6};

			\end{tikzpicture}
		\end{center}
	 
		By \eqref{g100},  \eqref{Numerator1} and \eqref{22g30}, one has 
		for  any $y$ with $|y-jq_n-\frac{1}{2} q_n|\leq 10 \varepsilon q_n$,
		\begin{align}
			|\phi(y)|\leq & |P_{[x_1+q_n,p-1+q_n]}|^{-1} |P_{[y+1,p+q_n]}| (r_{j} e^{-L  (\frac{q_n}{4}-x) +C\varepsilon q_n } +r_{j-\frac{1}{2}} e^{-L (\frac{q_n}{4} +x)+C\varepsilon q_n} )\nonumber\\
			&+|P_{[x_1+q_n,p+q_n]}|^{-1}|P_{[x_1+q_n,y-1]}|   r_{j+1}\nonumber\\
			\leq &  |P_{[x_1+q_n,p-1+q_n]}|^{-1}  e^{L \frac{q_n}{2}+C\varepsilon q_n} (r_{j} e^{-L  (\frac{q_n}{4}-x)  } +r_{j-\frac{1}{2}} e^{-L (\frac{q_n}{4} +x)} )\nonumber\\
			&+|P_{[x_1+q_n,p-1+q_n]}|^{-1}e^{\frac{3}{4}Lq_n-Lx +C\varepsilon q_n}  r_{j+1}.\label{g56}
		\end{align}

		This implies
		\begin{equation} \label{g58}
		r_{j+\frac{1}{2}}	|P_{[x_1+q_n,p-1+q_n]}| 
			\leq    e^{L \frac{q_n}{4}+C\varepsilon q_n} (r_{j} e^{Lx } +r_{j-\frac{1}{2}} e^{-Lx} )+  e^{\frac{3}{4} Lq_n-Lx+C\varepsilon q_n} r_{j+1}.
		\end{equation}
		%where the second inequality holds by \eqref{Numerator1}.

		%		By Lemmas \ref{klem2} and \ref{klem1}, one has
		%		\begin{align}
		%			|P_{[x_1+q_n,y]}|  \leq  &  |P_{[x_1+q_n-jq_n,y-jq_n]}| +e^{L(\frac{q_n}{2}+\frac{q_n}{4}-x)-\beta_{j}q_n+C\varepsilon q_n}\nonumber\\
		%			\leq &e^{L (\frac{q_n}{4}+x)+C\varepsilon q_n} +e^{L(\frac{3q_n}{4}-x)-\beta_{j}q_n+C\varepsilon q_n}.\label{g57}
		%		\end{align}
		By \eqref{g58}, one has 
		\begin{align}
			e^{-\frac{3}{2}L q_n +2\beta_{j} q_n}e^{-L(\frac{q_n}{4}-x)}& 	|P_{[x_1+q_n,p-1+q_n]}	|r_{j+\frac{1}{2}}\nonumber\\
			\leq&e^{- \frac{3}{2} L q_n+2\beta_{j}q_n+2L x+C\varepsilon q_n} r_j 
			+e^{-\frac{3}{2}L q_n+2\beta_{j} q_n+C\varepsilon q_n}r_{j-\frac{1}{2}}\nonumber\\
			&+e^{-  L q_n+2\beta_{j}q_n +C\varepsilon q_n}r_{j+1}\nonumber\\
			\leq & e^{-L q_n+2\beta_{j}q_n+C\varepsilon q_n}r_{j+1}+e^{- L q_n+2\beta_{j} q_n+C\varepsilon q_n}r_{j},\label{g60}
		\end{align}
		where the last inequality holds by the fact that $r_{j-\frac{1}{2}}\leq e^{\frac{q_n}{2}L+C\varepsilon q_n} r_j$  and $0\leq x\leq \frac{q_n}{4}$.
		
	%	Since \eqref{g59} holds for any $y$ with $|y-jq_n-\frac{1}{2}q_n|\leq 10 \varepsilon q_n$, one has
	%	\begin{equation}\label{g60}
		%	e^{-\frac{3}{2}L q_n +2\beta_{j} q_n}e^{-L(\frac{q_n}{4}-x)}	|P_{[x_1+q_n,p-1+q_n]}| r_{j+\frac{1}{2}}\leq  e^{-L q_n+2\beta_{j}q_n+C\varepsilon q_n}(r_{j+1}+r_{j}).
	%	\end{equation}
		By \eqref{g61} and \eqref{g60}, one has
		\begin{align}
			e^{-\frac{3}{2} L q_n+2\beta_{j} q_n}&|P_{[x_1,p-1]}|   e^{-L (\frac{q_n}{4}-x) } r_{j+\frac{1}{2}}\nonumber\\
			\leq &e^{-L q_n+2\beta_{j}q_n+C\varepsilon q_n}(r_{j+1}+r_j)+e^{- (L-2\beta_{j}-C\varepsilon)\frac{q_n}{2}}r_{j+\frac{1}{2}}. \label{g62}
		\end{align}

		By \eqref{g63}, \eqref{g51}, \eqref{g55} and \eqref{g62}
		\begin{align}
			r_j \leq &e^{-(L  -2\beta_{j}-C\varepsilon)q_n}   (r_{j-1}+r_{j+1}) +  e^{-(L-2\beta_{j})\frac{q_n}{2}+C\varepsilon q_n} r_{j+\frac{1}{2}}\nonumber\\
			&+(e^{-L q_n+2\beta_{j}q_n+C\varepsilon q_n}+e^{- \frac{1}{2}L q_n+\beta_{j} q_n+C\varepsilon q_n}) r_{j}.\label{g128}
		\end{align}
		Since 
		\begin{equation*}
			(e^{-L q_n+2\beta_{j}q_n+C\varepsilon q_n}+e^{- \frac{1}{2}L q_n+\beta_{j} q_n+C\varepsilon q_n}) r_{j}\leq \frac{1}{2}r_j,
		\end{equation*}
		by \eqref{g128}, one has
		\begin{equation}\label{g70}
			r_j\leq 
			e^{-(L  -2\beta_{j}-C\varepsilon)q_n}   (r_{j-1}+r_{j+1}) +  e^{-(L-2\beta_{j}-C\varepsilon)\frac{q_n}{2}} r_{j+\frac{1}{2}}.
		\end{equation}
		
		{\bf Case 4}:		$j_0\in[jq_n+\frac{1}{4}q_n, jq_n+\frac{1}{2}q_n]\cap I_2$
		
		In this case, following the proof of Case 3, we have 
		\begin{equation}\label{g71}
			r_j\leq e^{-(L  -2\beta_{j}-C\varepsilon)q_n}   (r_{j-1}+r_{j+1}) +  e^{-(L-2\beta_{j}-C\varepsilon)\frac{q_n}{2}} r_{j-\frac{1}{2}}.
		\end{equation}

		If  $j_0\in I_1$,  then  \eqref{g44}, \eqref{g45}, (\ref{g70})  and \eqref{g71} hold  for  $j=0$.   Therefore,   we have
		\begin{equation*}
			|\phi(0)|\leq \frac{1}{2}.
		\end{equation*}
		This is in contradiction with  $\phi(0)=1$.
		Therefore  we must have $j_0\in I_2$ and \eqref{g1000} follows from \eqref{g44}, \eqref{g45}, (\ref{g70})  and \eqref{g71}.
	\end{proof}
	
	\section{Proof of  Theorem \ref{mainthm}}\label{S6}
	
	Once we have Theorems \ref{thm1} and \ref{thm2} at hand,  Theorem \ref{mainthm} follows from standard iterations. 
	See ~\cite{jl18,liuetds20} for example.  For the convenience, we  include a proof here.
	\begin{proof}[\bf Proof of  Theorem \ref{mainthm}]
		
		Without loss of generality, we only bound $\phi(k)$ with $k>0$. Let $n$ be such that  $b_n\leq k<b_{n+1}$.
		Clearly, $\beta_j\leq \beta+\varepsilon$.
		By  Theorems \ref{thm1} and \ref{thm2},  we have for any $j$ with $ 1\leq   j  \leq 2\frac{b_{n+1}}{q_n}+10 $,  
		\begin{equation}\label{G.r_j2}
			r_{j-\frac{1}{2} }\leq  \exp\{ - \frac{1}{2}(L-2\beta-C\varepsilon) q_n \} \max\{r_{j-1},r_{j}\} ,
		\end{equation}
		and
		\begin{equation}\label{G.r_j1}
			r_{j }\leq  \max_{t\in O}\{\exp\{ -|t|(L-2\beta-C\varepsilon) q_n \} r_{j+t}\} ,
		\end{equation}
		where $O=\{ \pm1,\pm\frac{1}{2}$\}.
		
		Suppose  $1 \leq \ell \leq \frac{b_{n+1}}{q_n}+4 $.
		Let $j=\ell$ in   (\ref{G.r_j1}) and (\ref{G.r_j2}),   and iterate     $ 2\ell$  times or until $j\leq 1$,
		we obtain
		\begin{equation}\label{radd1}
			r_{\ell} \leq  (2\ell+2)q_n\exp\{- (L-2\beta -C\varepsilon) \ell q_n\},
		\end{equation}
		and
		\begin{equation}\label{radd2}
			r_{\ell-\frac{1}{2}}\leq  (2\ell+2)q_n\exp\{- (L-2\beta -C\varepsilon) (\ell -\frac{1}{2})q_n\}.
		\end{equation}
		Notice that  we have used the fact that $|r_{j}|\leq (j+1)q_n$ and $|r_{j-\frac{1}{2}}|\leq (j+1)q_n$.

		{\bf 	Case 1:}  $k\geq \frac{q_n}{4}$  and $  {\rm dist}(k, q_n\mathbb{Z}+\frac{q_n}{2}\mathbb{Z} )\leq     10\varepsilon q_n$.
		
		In this case, applying \eqref{radd1} and \eqref{radd2}, one has
		\begin{equation}\label{Decayoct1}
			| \phi(k)| ,  | \phi(k-1)|\leq   \exp\{- (L-2\beta -C\varepsilon)k\}.
		\end{equation}
		
		{\bf Case 2}:  others
		
		Applying  Lemma \ref{Le.resonant}  with sufficiently small $\varepsilon$,  and by \eqref{radd1} and \eqref{radd2}, 
		one also has
		\begin{equation}\label{Decayoct2}
			| \phi(k)|, | \phi(k-1)|\leq   \exp\{- (L-2\beta -C\varepsilon )k\}.
		\end{equation}
		We finish the proof.

	\end{proof}	
	
	\appendix
	\section{Proof of  Claims 1 and  2}
	Let $ \frac{p_n}{q_n}$ be the continued fraction approximations to $\alpha$, then
	\begin{equation}\label{GDC1}
		\forall 1\leq k <q_{n+1},  \text{dist}( k\alpha,\mathbb{Z})\geq  |q_n\alpha -p_n|,
	\end{equation}
	and
	\begin{equation}\label{GDC2}
		\frac{1}{2q_{n+1}}\leq|q_n\alpha-p_n| \leq\frac{1}{q_{n+1}}.
	\end{equation}
	
	\begin{proof}[\bf Proof of Claim 1]
		
		By the construction of $I_1$ and $I_2$ in Claim 1,   \eqref{GDC1} and \eqref{GDC2},
		we have  that
		\begin{itemize}
			\item  for any $m \in I_1$,
			\begin{equation}\label{Appenoct1}
				\min_{\ell \in I_1\cup I_2}\ln |\sin \pi(2\theta +(\ell +  m)\alpha)|\geq -C\ln q_n,
			\end{equation}
			and
			\begin{equation}\label{Appenoct2}
				\min_{\ell \neq m\atop \ell \in I_1\cup I_2} \ln |\sin
				\pi (\ell-m)\alpha) | \geq -C\ln q_n;
			\end{equation}
			\item 
			for any $m\in I_2$,
			\begin{equation}\label{Appenoct3}
				\min_{\ell \in I_1\cup I_2}\ln |\sin \pi(2\theta +(\ell +  m)\alpha)|_{\mathbb{R}/\mathbb{Z}}\geq  -\beta_{j} q_n-C\ln q_n,
			\end{equation}
			and
			\begin{equation}\label{Appenoct4}
				\min_{\ell \neq m\atop \ell \in I_1\cup I_2} \ln |\sin
				\pi (\ell-m )\alpha) | \geq -C\ln q_n.
			\end{equation}
		\end{itemize}
		
		We should mention that, for each $m\in I_2$, there is at most one $ \ell\in I_1\cup I_2$ such that the lower bound of \eqref{Appenoct3} can be achieved.
		
		Once we have \eqref{Appenoct1}-\eqref{Appenoct4} at hand, 
		by the standard arguments (e.g.  Appendices in ~\cite{jl18,liuetds20}), we have that  for any $m\in I_1$,
		\begin{equation*}
			{\rm Lag}_m\leq \varepsilon q_n 
		\end{equation*}
		and for any $m\in I_2$,
		\begin{equation*}
			{\rm Lag}_m\leq \beta_j q_n+ \varepsilon q_n.
		\end{equation*}
	\end{proof}
	
	\begin{proof}[\bf Proof of Claim 2]
		
		By the construction of $I_1$ and  $I_2$  in Claim 2,   \eqref{GDC1} and \eqref{GDC2},
		we have  that for $m \in I_1 \cup [jq_n-2s q_{n-n_0},jq_n-1]$,
		\begin{equation}\label{Appenoct11}
			\min_{\ell \in I_1\cup I_2 }\ln |\sin \pi(2\theta +(\ell +  m)\alpha)|\geq \beta_{j} q_n-C\ln q_n,
		\end{equation}
		and
		\begin{equation}\label{Appenoct21}
			\min_{\ell \neq m\atop \ell \in I_1\cup I_2} \ln |\sin
			\pi (\ell-m)\alpha) | \geq  \beta_{j} q_n-C\ln q_n.
		\end{equation}
		We should mention that, for each $m \in I_1 \cup [jq_n-2s q_{n-n_0},jq_n-1]$, there is at most one $ \ell \in I_1\cup I_2$ such that the lower bound of \eqref{Appenoct11} or  \eqref{Appenoct21}  can be achieved.
		
			We also  have  that for $m \in[jq_n, jq_n+2 sq_{n-n_0}-1]$,
		\begin{equation}\label{Appenoct111}
		\min_{\ell \in I_1\cup I_2 }\ln |\sin \pi(2\theta +(\ell +  m)\alpha)|\geq \beta_{j} q_n-C\ln q_n,
		\end{equation}
		and
		\begin{equation}\label{Appenoct211}
		\min_{\ell \neq m\atop \ell \in I_1\cup I_2} \ln |\sin
		\pi (\ell-m)\alpha) | \geq   -C\ln q_n.
		\end{equation}
		Moreover,  for  each $m \in[jq_n, jq_n+2 sq_{n-n_0}-1]$,   there is at most two $ \ell \in I_1\cup I_2$ such that the lower bound of \eqref{Appenoct111}.
		Once we have \eqref{Appenoct11}-\eqref{Appenoct211} at hand, 
		by the standard arguments (e.g. Appendices in ~\cite{jl18,liuetds20}),
		we have  that for any $m\in I_1\cup I_2$,
		\begin{equation*}
			{\rm Lag}_m\leq 2\beta_{j}q_n+ \varepsilon q_n.
		\end{equation*}
		
	\end{proof}

	\section*{Acknowledgments}
	%The authors are very grateful to the anonymous referees for their knowledgeable reports, which helped us to improve our manuscript.
	The author  was  supported by   NSF DMS-2000345 and  DMS-2052572.
	
	%his work was done when W. Liu and S. Jitomirskaya visited Isaac Newton Institute for Mathematical Sciences in Cambridge.
	% We both thank the invitation of organizers of Periodic and Ergodic Spectral Problems.
	%Department of Mathematics, the University of California at
	%Irvine for warm hospitality, especially thank Professor Martin Schechter for constant encouragement and helpful discussions.

\end{document}